\documentclass[11pt]{article}

\usepackage[preprint]{acl}

\usepackage{times}
\usepackage{latexsym}

\usepackage[T1]{fontenc}
\usepackage[utf8]{inputenc}

\usepackage{microtype}

\usepackage{inconsolata}

\usepackage{graphicx}

\usepackage{amsmath}
\usepackage{amssymb}
\usepackage{amsthm}

\usepackage{booktabs}
\usepackage{arydshln}
\usepackage{multirow}
\usepackage{array}
\usepackage{tabularx}
\usepackage{float}
\usepackage{url}
\usepackage{hyperref}
\usepackage{cleveref}
\usepackage{algorithm}
\usepackage{algorithmic}

\newtheorem{theorem}{Theorem}

\newtheorem{proposition}{Proposition}
\newtheorem{corollary}{Corollary}
\newtheorem{remark}{Remark}

\newcommand{\rev}[1]{#1}

\newcommand{\methodours}{CROSS-MAP}
\newcommand{\methodhas}{HAS}

\title{LLMs as Linguistic Chameleons: Decoupling Semantics and Structure for Privacy-Preserving Communication}

\author{
  Yuzhu Mao \\
  Emory University \\
  \texttt{yuzhu.mao@emory.edu} \\\And
  Liang Zhao\thanks{Corresponding author.} \\
  Emory University \\
  \texttt{liang.zhao@emory.edu} \\
}

\begin{document}
\maketitle
\begin{abstract}
% Third-party large language models (\LLM s) offer strong reasoning and generation capabilities, but user prompts often contain sensitive personal or proprietary information. We study a recoverable privacy mechanism that conceals sensitive intent by \emph{topic shifting}: a local \LLM maps a sensitive prompt into a semantically consistent but domain-shifted surrogate, queries a third-party \LLM with only the surrogate, and then locally restores the answer back to the original domain using a retained mapping dictionary. We propose a practical training pipeline from supervised fine-tuning (\SFT) to direct preference optimization (\DPO) \citep{rafailov2023dpo} that improves structured-output validity, fluency, privacy, and recoverability. We introduce a metric suite that quantifies (i) dictionary coverage, (ii) semantic privacy at entity and sentence levels, (iii) fluency, and (iv) restoration quality; these signals construct preference pairs for \DPO. Finally, we evaluate instance-level privacy under non-cooperative inversion attacks where an adversary intercepts only the topic-shifted text and attempts to reconstruct the original content without access to the dictionary. Experiments on instruction-following \LLM s demonstrate that \DPO yields a better privacy--recoverability trade-off than \SFT alone while maintaining stable structured outputs required by restoration.

\end{abstract}
As Large Language Model (LLM) APIs become increasingly integrated into privacy-sensitive workflows, ensuring inference-time privacy without compromising task utility remains a major challenge. Existing approaches preserve most of the original semantic content to maintain downstream performance, but this also leaves exploitable cues for reconstructing the original text. This work investigates semantic decoupling, which replaces original semantics with alternative content while preserving the structure needed for LLM reasoning. Based on this idea, we propose CROSS-MAP, a bidirectional framework that maps private inputs into a different semantic domain before inference and recovers the corresponding outputs afterward. Local models are trained with multi-objective optimization to maximize semantic divergence in the mapping stage while minimizing semantic inconsistency in the recovery stage. Experiments show that CROSS-MAP reduces reconstruction success across multiple attack settings while outperforming existing baselines in utility. Code is available at \url{https://anonymous.4open.science/r/CROSS-MAP-7B2C/}.

\section{Introduction}
% {\color{blue}
As LLM APIs are increasingly integrated into private workflows, a growing amount of sensitive communication, from user queries to inter-agent instructions, is sent to third-party providers. Protecting textual information has therefore become a critical challenge in the era of LLMs~\cite{hongdp,kan2023protectinguserprivacyremote}. A standard view in linguistics and natural language processing (NLP) is that text decomposes into semantics (meaning-bearing elements such as entities, events, and attributes) and structure (the syntactic organization of these elements, such as word order and dependency relations)~\cite{Bloom1979LanguageDA,Hill2025DLD}. Many reasoning tasks in LLMs are largely driven by structural patterns, implying that logical relations among entities remain valid regardless of changes in semantic content~\cite{kim-linzen-2020-cogs, li-etal-2023-slog}. Since sensitive information is encoded primarily in semantic content rather than abstract structural templates, privacy leakage typically arises from semantics~\cite{frikha2025incognitext,staabbeyond}, while structure remains largely benign.

Existing privacy-preserving approaches for LLM interaction face a fundamental trade-off between privacy and utility. To maintain downstream task performance, these methods typically keep the transformed text semantically close to the original input, since inference is performed directly on the transmitted text. Methods such as differential privacy and anonymization obfuscate sensitive information by perturbing semantic content, often degrading utility for tasks that rely on fine-grained details~\citep{wuprivacy,hongdp,staab2025large}. Other approaches introduce recovery mechanisms to mitigate semantic drift~\cite{chowdhury2025pr,kan2023protectinguserprivacyremote,shen2024fire,chen2023hideseekhaslightweight}. However, since these methods operate in a semantic preserving regime, the semantic content is subject to only limited changes to retain utility, leaving exploitable cues that allow attackers to infer or reconstruct the original text~\citep{tong2025vulnerability,reconstructiondp}. For example, even after transformation, an LLM can still infer a user’s sensitive financial and medical circumstances from a statement such as \textit{Alice transferred $\$50,000$ from her savings account to pay medical debt after being diagnosed with serious disease}.

To address the dilemma inherent in semantic preserving approaches, we explore a new scenario, \textbf{semantic decoupling}, that allows LLMs to reason over preserved structural patterns while sensitive semantics remain local. Instead of transmitting sensitive content directly, the input is mapped into an \textbf{insensitive semantic domain} that preserves the structure but replaces the meaning. This allows the LLM to operate on a structure-preserving representation, where the original semantics is completely hidden and can only be restored locally afterward. For example, \textit{the thief robs the bank}, \textit{the waiter wipes the table}, and \textit{the programmer tests the software} share the same \textit{noun + verb + determiner + noun} structure despite expressing different meanings. Thus, for queries such as identifying who performs which action, the LLM can reason over the preserved structure, while only an authorized receiver with the mapping codebook (e.g., \textit{thief $\leftrightarrow$ waiter, rob $\leftrightarrow$ wipe, bank $\leftrightarrow$ table}) can recover the original text, leaving attackers observing the mapped text with little information about the sensitive content.

Historically, this paradigm has been difficult to realize. Semantic decoupling requires discovering mappings across semantic domains that preserve structural reasoning while enabling reliable recovery, which in effect requires constructing a codebook that aligns concepts across domains. Identifying such mappings demands broad linguistic and conceptual knowledge, which was largely infeasible before the emergence of modern foundation models. Recent LLMs show a strong ability to model cross-domain semantic relationships, making automated mapping increasingly plausible. However, for privacy, this capability must run locally rather than through a third-party API, since sending sensitive inputs externally would defeat the purpose of semantic decoupling. Meanwhile, general foundation models are often too large for practical local deployment and are not trained for semantic-decoupling mapping and recovery. A practical solution must therefore meet three requirements simultaneously: efficient local deployment, strict locality to avoid external exposure of private inputs, and adaptability to semantic-decoupling tasks.

This paper explores the extent to which semantic decoupling, which replaces sensitive semantic content while preserving structural reasoning patterns, can move beyond the privacy–utility dilemma inherent in semantics-preserving approaches. Specifically, we first formulate the semantics-preserving and semantic-decoupling paradigms and analyze their respective advantages from an information-theoretic perspective, characterizing when semantic decoupling offers a better privacy–utility trade-off. Second, we propose CROSS-MAP, a compact local framework for semantic-decoupling text mapping and recovery. With training objectives that encourage large semantic shifts while maintaining consistent recovery, compact local models such as Qwen-2.5-7B can be adapted to perform semantic decoupling through multi-objective optimization, achieving both stronger privacy and higher utility than semantics-preserving baselines. Finally, in addition to the defense-centric view, we adopt an attack-centric view of security by introducing an optimization-based attack to evaluate the robustness of semantic-decoupling systems against white-box adversaries, including scenarios with partial codebook leakage.
% }

% \vspace{-10pt}
\section{Related Work}
This section provides a brief overview of related work. A more detailed review is provided in Appendix~\ref{sec:appendix_related_work}.

\subsection{Differential Privacy}

Differential Privacy (DP) approaches for LLMs inject noise into data, gradients, or generation processes to avoid directly exposing the original private dataset to an external model~\citep{abadi2016deep,wuprivately,yue2023synthetic,wang2025rewardds,wuprivacy,utpala2023locally,flemings2024differentially,hongdp}. However, DP noise can significantly degrade model performance. To preserve downstream utility, existing methods often either restrict DP protection to less critical components~\citep{chowdhury2025pr,tong2025inferdpt} or focus on tasks that are relatively insensitive to fine-grained details, such as sentiment analysis~\citep{mattern2022differentially,kurakin2023harnessing}.

\vspace{-5pt}
\subsection{Text Anonymization and Sanitization}
Text anonymization aims to transform text to reduce the risk of revealing personal information (e.g., identity or sensitive attributes) while preserving downstream utility~\citep{pilan2022text}. Prior work demonstrates that modern LLMs have the potential to both anonymize and deanonymize text, highlighting risks of malicious re-identification attacks~\citep{patsakis2023man,staabbeyond,staab2025large}. Text sanitization extends beyond protecting personal identity information to broader sensitive content~\citep{yue2021differential,li2025papillon,chen2023customized}. To further improve downstream utility, some studies introduce an explicit desanitization step to form a bidirectional framework~\citep{kan2023protectinguserprivacyremote, shen2024fire,chen2023hideseekhaslightweight,chowdhury2025pr}. This line of works share a consensus that the mapped text and original text must be within the same semantic space for more accurate LLM inference.

\section{Preliminaries}
This section presents necessary knowledge for understanding the difference between the proposed method and related work.

% \subsection{Background}
% Framework that contains both mapping and recovery processes is termed as \textbf{bidirectional privacy-preserving framework} in this work. Let $\mathcal{X}$ be the space of original texts and $\mathcal{Y}$ the space of mapped texts. A mapping process $\mathcal{M}$ produces a mapped output $y \sim \mathcal{M}(x)$ from the original text $x$. A recovery process $\mathcal{R}$ produces a recovered text $\tilde{x} \sim \mathcal{R}(y)$ from $y$.

% \vspace{-5pt}
% \paragraph{Threat model.}
% Attacks for recovering the original context $x$ under the mapping mechanism can be categorized into black-box and white-box types. In a black-box attack, the attacker only observes the mapped output $y$ and attempts to infer the original content by producing a guess $\hat{x} \sim \mathcal{A}(y)$, where $\mathcal{A}$ denotes an attack algorithm that maps the observation $y$ to candidate reconstructions $\hat{x}$. In a white-box attack, the attacker additionally knows the mapping mechanism $\mathcal{M}$. Therefore, instead of relying purely on heuristic guessing, the attacker can actively search for inputs $\hat{x}$ whose mapped outputs are consistent with the observation $y$.

\subsection{Problem Formulation}

Framework that contains both mapping and recovery processes is termed as \textbf{bidirectional privacy-preserving framework} in this work. Let $\mathcal{X}$ be the space of original texts and $\mathcal{Y}$ the space of mapped texts. A mapping process $\mathcal{M}$ produces a mapped output $y \sim \mathcal{M}(x)$ from the original text $x$. A recovery process $\mathcal{R}$ produces a recovered text $\tilde{x} \sim \mathcal{R}(y)$ from $y$. Privacy-preserving communication with LLMs aims to maintain the downstream task utility while protecting the original input $x$. Without protection, the LLM outputs a model response $x_{\mathrm{a}}=\mathcal{LLM}(x)$ upon receiving the input $x$. But to prevent the LLM from directly observing the original input $x$, the mapping mechanism $\mathcal{M}$ first conducts the mapping $x \mapsto y$ and then sends $y \in \mathcal{Y}$ to the LLM. Given the mapped input $y$, the LLM outputs a model response $y_{\mathrm{a}}=\mathcal{LLM}(y)$.

After receiving the model response $y_{\mathrm{a}}$ from the LLM, the recovery mechanism $\mathcal{R}$ produces a recovered response $\tilde{x}_{\mathrm{a}}=\mathcal{R}(y_{\mathrm{a}})$.The downstream task utility requires the recovered response $\tilde{x}_{\mathrm{a}}$ to be consistent with the ground-truth response $x_{\mathrm{a}}$:
\vspace{-6pt}
\[
\mathrm{Sim}(x_{\mathrm{a}}, \tilde{x}_{\mathrm{a}}) \ge 1-\epsilon,
\vspace{-6pt}
\]
where $\mathrm{Sim}(\cdot,\cdot)$ is a text similarity measure, e.g., embedding cosine, BERTScore, or an LLM-judge. The threshold $\epsilon$ is a measure for the drift in the model response induced by the mapping and recovery processes.

\subsection{Two Mapping Paradigms}
According to linguistic research, language is comprised of three primary components: form, content, and use~\cite{Bloom1979LanguageDA,Hill2025DLD}. Form covers syntax, morphology, and orthography. Content denotes semantics, while use refers to pragmatics, which relates to the context of communication. Consequently, any written text $x$ itself can be decomposed into structural component $x_c$ (e.g., sentence organization, word worder, and relationships between words) and semantic element $x_s$ (e.g., entities, phrases, and key spans). 

% Under such decomposition, the similarity between text $x$ and $x'$ can be decomposed into the similarity between the structural components $x_c$ and $x'_c$ and the similarity between the semantic elements $x_s$ and $x'_s$:
% \[
%   \begin{aligned}
%   \mathrm{Sim}(x,x')
% = &\quad \alpha\,\mathrm{Sim}_{\text{semamtic}} \!\left(x_s,x'_s\right)
% \\
% & + (1-\alpha)\,\mathrm{Sim}_{\text{structural}} \!\left(x_c,x'_c\right),
%   \end{aligned}
% \]
% where $\alpha\in[0,1]$ controls the trade-off between semantic and structural similarity.

Under such decomposition, we present the key difference between existing mapping mechanisms, which is summarized as \textit{semantic-preserving mapping}, and the proposed \textit{semantic-decoupling mapping}: The goal of \textit{semantic-preserving mapping} can be formulated as solving the following optimization problem
\vspace{-8pt}
\[
\max_{\mathcal{M},\mathcal{R}}
\; \mathrm{Sim} \!\left(\mathcal{LLM}(x),\;\mathcal{R}\!\left(\mathcal{LLM}(\mathcal{M}(x))\right)\right),
\vspace{-6pt}
\]
subject to the \textit{semantic similarity constraint}:
\vspace{-6pt}
\[
\mathrm{Sim} \!\left(x_s,\;\mathcal{M}(x_s)\right) \ge 1 - \epsilon_{s}.
\]

\vspace{-6pt}
In contrast, the goal of \textit{semantic-decoupling mapping} can be formulated as subject to the \textit{structural similarity constraint}:
\vspace{-6pt}
\[
\mathrm{Sim} \!\left(x_c,\;\mathcal{M}(x_c)\right) \ge 1 - \epsilon_{c}.
\vspace{-6pt}
\]

% \vspace{-7pt}
% \[
% \max_{\mathcal{M},\mathcal{R}}
% \; \mathrm{Sim} \!\left(\mathcal{LLM}(x),\;\mathcal{R}\!\left(\mathcal{LLM}(\mathcal{M}(x))\right)\right),
% \]

% subject to the \textit{structural similarity constraint}:
% \[
% \mathrm{Sim}_{\text{structural}} \!\left(x_c,\;\mathcal{M}(x_c)\right) \ge \epsilon_{c}.
% \]

% \paragraph{Mapping functions.}
% A mapping mechanism $M$ transforms an original text $x$ into a mapped text $y$.
% Using the structural decomposition above, the mapping can be written as
% \[
% M(x) = \big(M_c(c(x)),\,M_h(h(x))\big),
% \]
% where $M_c$ operates on structure and $M_h$ operates on lexical–semantic
% content.

Briefly speaking, both mapping paradigms aim to maximize the similarity between the ground-truth LLM response $x_{\mathrm{a}}=\mathcal{LLM}(x)$ and the recovered response $\tilde{x}_{a}=\mathcal{R}(\mathcal{LLM}(\mathcal{M}(x)))$. \textit{Semantic-preserving mapping} requires the mapped text $\mathcal{M}(x)$ to remain semantically similar to the original text $x$. For example, \textit{the number increases rapidly} $\mapsto$ \textit{the value rises quickly}. However, the \textit{semantic-decoupling mapping} removes this constraint on semantic similarity. Instead, it requires structure invariance between the original text $x$ and the mapped text $\mathcal{M}(x)$. For example, \textit{The number increases rapidly} $\mapsto$ \textit{The rain drops heavily}.

% In the next section, the goal of these two mapping paradigms is theoretically analyzed via an information bottleneck formulation.

% =========================
% Theory
% =========================
\section{An Information-Theoretic View}
\label{sec:theory}
The difference between \textit{semantic-preserving} and \textit{semantic-decoupling} mappings can be understood from an information-theoretic perspective. Recall the mapping pipeline $(x_c, x_s) \xrightarrow{\mathcal{M}} (y_c, y_s) \xrightarrow{\mathcal{LLM}} z$, where $x_c$ and $x_s$ are the structural and semantic components of the original text $x$, $y_c$ and $y_s$ are the structural and semantic components of the mapped text $y$, and $z$ is the LLM response on $y$.

\vspace{-6pt}
\paragraph{Information bottleneck objective.} Searching for the optimal mapping mechanism $\mathcal{M}$ is equivalent to solving the following optimization problem:
\vspace{-8pt}
% The mapping process can be interpreted as solving an optimization problem governed by the information bottleneck objective
\begin{equation}
\mathcal{G}(\mathcal{M}) = \min_{\mathcal{M}} I(x;y) - \beta I(y;z),
\vspace{-6pt}
\label{eq:ib_objective}
\end{equation}
where $I(x;y)$ quantifies the amount of information from the original text retained within the mapped text, and $I(y;z)$ measures the task-relevant information preserved in the mapped text for the downstream LLM. \rev{For readability, Table~\ref{tab:core_notation} summarizes the core notation used in the main text.}

\begin{table}[t]
\centering
\small
\begin{tabularx}{\linewidth}{lX}
\toprule
\rev{Symbol} & \rev{Meaning} \\
\midrule
\rev{$x, y, z$} & \rev{original text, mapped text, and LLM response} \\
\rev{$x_s, x_c$} & \rev{semantic and structural components of $x$} \\
% \rev{$d_{\mathrm{text}}$} & \rev{text-level text distance} \\
% \rev{$d_{\mathrm{dict}}$} & \rev{dictionary-level text distance} \\
\rev{$\lambda_s, \lambda_c$} & \rev{task dependence on semantic and structural information, respectively} \\
\rev{$\mathcal{M}_{\mathrm{SD}}, \mathcal{M}_{\mathrm{SP}}$} & \rev{semantic-decoupling~/semantic-preserving mapping mechanisms} \\
\rev{$\Delta\rho_s, \Delta\rho_c$} & \rev{how much less semantic and structural information from the original text is preserved by $\mathcal{M}_{\mathrm{SD}}$ than by $\mathcal{M}_{\mathrm{SP}}$ in the mapped text (i.e., $\mathcal{M}_{\mathrm{SD}}$'s compression gain)} \\
\rev{$\Delta\eta_s, \Delta\eta_c$} & \rev{how much less task-relevant semantic and structural information is preserved by $\mathcal{M}_{\mathrm{SD}}$ than by $\mathcal{M}_{\mathrm{SP}}$ in the mapped text (i.e., $\mathcal{M}_{\mathrm{SD}}$'s task-relevant information loss)} \\
\bottomrule
\end{tabularx}
\vspace{-5pt}
\caption{\rev{Core notation used in the main text. Other notations used in the proof are defined in Appendix~\ref{sec:proof_ib_sd_advantage}.}}
\label{tab:core_notation}
\vspace{-15pt}
\end{table}

\begin{theorem}
  \label{thm:ib_sd_advantage}
  Let $\mathcal{M}_{\mathrm{SP}}$ and $\mathcal{M}_{\mathrm{SD}}$
  denote semantic-preserving and semantic-decoupling mappings
  respectively under the same task distribution over inputs
  $x$. Assume semantic and structural information is indepdent in any text. Then, the information bottleneck objective satisfies
  \vspace{-8pt}
  \[
  \mathcal{G}(\mathcal{M}_{\mathrm{SD}}) < \mathcal{G}(\mathcal{M}_{\mathrm{SP}}),
  \vspace{-5pt}
  \]
  % \begin{align*}
  % &\big(I^{\mathcal{M}_{\mathrm{SD}}}(x;y)-\beta I^{\mathcal{M}_{\mathrm{SD}}}(y;z)\big)
  % \\
  % \le &\big(I^{\mathcal{M}_{\mathrm{SP}}}(x;y)-\beta I^{\mathcal{M}_{\mathrm{SP}}}(y;z)\big),
  % \label{eq:ib_order}
  % \end{align*}
  i.e., $\mathcal{M}_{\mathrm{SD}}$ achieves a lower
  information bottleneck objective than $\mathcal{M}_{\mathrm{SP}}$, if the downstream task's semantic-dependence coefficient $\lambda_s$ is smaller than this threshold:
  \begin{equation}
  \lambda_{s} < \frac{r \cdot \Delta \rho_{s} + (1 - r) \Delta \rho_{c} - \beta' \cdot \Delta \eta_{c}}{\beta' (\Delta \eta_{s} - \Delta \eta_{c})},
  \label{eq:ib_condition}
  \vspace{-5pt}
  \end{equation}
  where $\beta'$ and $r$ are some constants. The detailed definitions of involved quantities and a full derivation are provided in Appendix~\ref{sec:proof_ib_sd_advantage}.
  \end{theorem}

\vspace{-8pt}
\begin{corollary} $\mathcal{M}_{\mathrm{SD}}$ achieves a lower
  information bottleneck objective than $\mathcal{M}_{\mathrm{SP}}$, if its gains in compression outweigh its task-relevant information loss in the following way:
  \vspace{-8pt}
  \[
    r \cdot \Delta \rho_s + (1 - r) \Delta \rho_c > \beta' (\Delta \eta_s \cdot \lambda_s + \Delta \eta_c \cdot \lambda_c).
  \vspace{-5pt}
  \]
\end{corollary}

% {\color{blue}
\vspace{-6pt}
\begin{remark}
Eq.~\eqref{eq:ib_condition} formalizes an intuitive trade-off that $\mathcal{M}_{\mathrm{SD}}$ is advantageous when it brings substantial compression gains while only mildly harming the task-relevant information needed by the downstream task: In the numerator, larger $\Delta \rho_s$ or $\Delta \rho_c$ means that $\mathcal{M}_{\mathrm{SD}}$ removes more semantic or structural information from the original text than $\mathcal{M}_{\mathrm{SP}}$, leading to a stronger compression and privacy advantage. This enlarges the threshold, so $\mathcal{M}_{\mathrm{SD}}$ remains preferable for a wider range of tasks. The weight $r$ reflects where most of the entropy in the original input lies. When $r$ is larger, the information carried by the original text is dominated by its semantic content.

By contrast, larger $\Delta \eta_s$ means that semantic decoupling discards more task-relevant semantic information, which increases the denominator and thus lowers the threshold. Larger $\Delta \eta_c$ has a similar effect through the term $-\beta' \Delta \eta_c$ in the numerator. Finally, a larger $\beta'$ implies that downstream utility is weighted more heavily relative to compression, so the threshold becomes more stringent. 
\end{remark}
% }

\paragraph{Privacy benefits with compression.} According to the variational information bottleneck (VIB) framework, the mutual information terms in Equation~\eqref{eq:ib_objective} can be bounded as follows:
\vspace{-5pt}
\[
I(x; y) = \mathbb{E}_{y \sim p(y)} \left[ {KL} ( p(x|y) \parallel p(x) ) \right],
\vspace{-5pt}
\]
and
\vspace{-5pt}
\[
I(y; z) \geq \mathbb{E}_{y, z} \left[ \log q(z|y) \right],
\vspace{-4pt}
\]
where $q(z|y)$ is an auxiliary variational distribution used to approximate the posterior $p(z|y)$. On one hand, $\mathcal{M}_{\mathrm{SD}}$ reduces $I(x; y)$ by decoupling the semantic information of $x$ and $y$, which encourages smaller KL divergence between $p(y|x)$ and $p(y)$. Equivalently, the attacker posterior $p(x|y)$ becomes closer to the true prior $p(x)$ on average, implying that the mapped output $y$ reveals less usable information about the original text $x$. Therefore, this compression directly enhances privacy against attackers observing $y$. \footnote{An in-depth analysis of the connection between KL divergence and posterior-based attack success rate is provided in Appendix~\ref{sec:appendix_asr}.}

On the other hand, such privacy benefits of $\mathcal{M}_{\mathrm{SD}}$ come at the cost of dropping more task-relevant semantic information compared to $\mathcal{M}_{\mathrm{SP}}$, which might impact the $q(z|y)$ modelled by the downstream LLM. As implied by Equation~\eqref{eq:ib_condition}, if the downstream task's dependence on semantic information can tolerate such loss, $\mathcal{M}_{\mathrm{SD}}$ dominates $\mathcal{M}_{\mathrm{SP}}$ with compression benefits outweighing the task information loss. 

\section{CROSS-MAP}\label{sec:method}
This section presents CROSS-MAP, a semantic-decoupling framework designed for cross-domain mapping and recovery.

\begin{figure*}[t]
  \centering
  \vspace{-18pt}
  \includegraphics[width=0.9\textwidth]{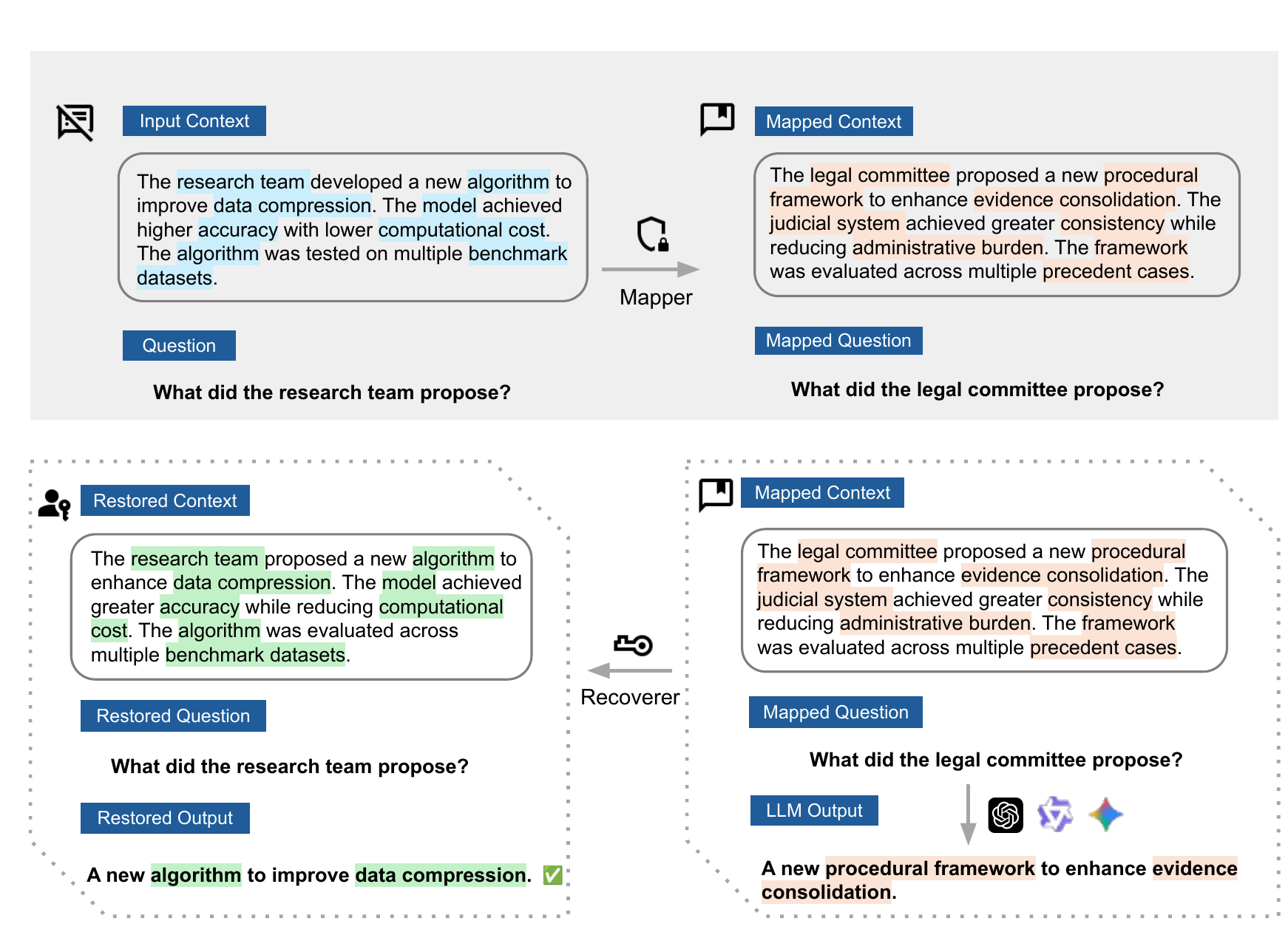}
  \vspace{-8pt}
  \caption{Illustration of inference-time protection with CROSS-MAP.}
  \label{fig:pipeline}
  \vspace{-15pt}
\end{figure*}

\subsection{Key Components}\label{subsec:key_components}
To address the challenges of highly-unstructured text inputs, two key designs are needed: (i) a structured mapping and a recovery module, and (ii) a verifiable reward and metric design.

\paragraph{Structured mapping and recovery.}\label{subsec:structured_mr}
Given an original input $x$, the mapper $\mathcal{M}_\phi$ produces a
cross-domain representation $(\tau, D, y) \sim \mathcal{M}_\phi(x)$ consisting of a target-domain theme $\tau$, a dictionary $D$,
and a mapped text $y$. Here
$D=\{(s_i,\hat s_i)\}_{i=1}^{m}$ maps each original span $s_i$ to a target-domain span $\hat s_i$. The recoverer $\mathcal{R}_\psi$ reconstructs the restored text $\tilde{x} \sim \mathcal{R}_\psi(\tau,D,y)$ from $(\tau,D,y)$. Both $\mathcal{M}_\phi$ and $\mathcal{R}_\psi$ are implemented as instruction-following LLMs with structured outputs
(e.g., JSON).
% The mapped text $y$ is written in the semantic space induced by $\tau$ and should apply $D$ consistently.

\vspace{-5pt}
\paragraph{Verifiable reward.}
A compact set of score metrics is needed to guide the training of the mapper and the recoverer:
\vspace{-6pt}
\begin{itemize}
  \item \textbf{Dictionary coverage $c_{\mathrm{dict}}(x,D)$.}
  Let $x=(x_1,\ldots,x_n)$ be the token sequence of length $n$. Each original span
  $s_i$ in $D$ corresponds to a contiguous index interval $I_i \subseteq \{1,\ldots,n\}$ in $x$.
  Define the set of covered token positions as $\mathcal{I}(D) := \bigcup_{i=1}^{m} I_i$. The dictionary coverage is the fraction of tokens in $x$ that are covered by a dictionary span: $c_{\mathrm{dict}}(x,D)
  =|\mathcal{I}(D)|/n$.
  Intuitively, higher $c_{\mathrm{dict}}$ means a larger portion of the original text $x$ is explicitly handled by the mapper through $D$.
  \vspace{-5pt}
  \item \textbf{Dictionary distance $d_{\mathrm{dict}}(D)$.}
  Using a cosine similarity measure $\mathrm{Sim}(\cdot,\cdot)$ that takes in span embeddings, the distance between spans in $D$ is measured as $d_{\mathrm{dict}}(D) = 1 - \frac{1}{m}\sum_{i=1}^{m}\mathrm{Sim}(s_i,\hat{s}_i)$.

  \item \textbf{Text distance $d_{\mathrm{text}}(x,y)$.} Using a cosine similarity measure that takes in sentence embeddings, the text distance between $x$ and $y$ is measured as $d_{\mathrm{text}}(x,y) = 1 - \frac{1}{L}\sum_{j=1}^{L}\mathrm{Sim}(x_{j},y_{j})$, where $\{x_{j}\}_{j=1}^{L}$ and $\{y_{j}\}_{j=1}^{L}$ are sentence splits of $x$ and $y$.
  % \item \textbf{Recovery quality $r_{\mathrm{rec}}(x,\tilde{x})$.} The recovery quality is measured as $r_{\mathrm{rec}}(x,\tilde{x}) = \frac{1}{L}\sum_{j=1}^{L}\mathrm{Sim}(x^{(j)},\tilde{x}^{(j)})$, where $\{\tilde{x}^{(j)}\}_{j=1}^{L}$ is the sentence split of $\tilde{x}$ aligned with $\{x^{(j)}\}$.
  % \begin{equation*}
  %   s_{\mathrm{r}}(x,\tilde{x})
  %   :=
  %   \frac{1}{L}\sum_{j=1}^{L}\mathrm{Sim}\!\big(x^{(j)},\tilde{x}^{(j)}\big).
  % \end{equation*}
  
  \vspace{-4pt}
  \item \textbf{Fluency $f_{\mathrm{text}}(y)$.}
  Let $(y_1,\ldots,y_T)$ be the token sequence of the mapped text $y$. The fluency score is set to the negative log-perplexity $-\frac{1}{T}\sum_{t=1}^{T}\log p_{\mathrm{LM}}\!\big(y_t | y_{<t}\big)$, where $p_{\mathrm{LM}}$ represents the output distribution of a reference language model.
  % The perplexity of $y$ under a reference LM $p_{\mathrm{LM}}$ is:
  % \begin{equation*}
  % \begin{split}
  %   \mathrm{PPL}(y):=\exp\!\bigg( -\frac{1}{T}\sum_{t=1}^{T}\log p_{\mathrm{LM}}\!\big(y_t \mid y_{<t}\big)\bigg),
  % \end{split}
  % \end{equation*}
  % and the fluency score is set to the negative log-perplexity:
  % \begin{equation*}
  % \begin{aligned}
  %   s_{\mathrm{f}}(y)
  %   &:= -\log \mathrm{PPL}(y) \\
  %   &= \frac{1}{T}\sum_{t=1}^{T}\log p_{\mathrm{LM}}\!\big(y_t \mid y_{<t}\big).
  % \end{aligned}
  % \end{equation*}
\end{itemize}

\vspace{-4pt}
Finally, the above metrics are aggregated into a single scalar reward:
\vspace{-6pt}
\begin{equation}
  \nonumber
\begin{aligned}
  \label{eq:overall_reward}
r(x, \tau,D,y,\tilde{x})
&= w_{1}\, c_{\mathrm{dict}}(x,D)
   + w_{2}\, d_{\mathrm{dict}}(D)
   \\
& + w_{3}\, d_{\mathrm{text}}(x,y)
   + w_{4}\, f_{\mathrm{text}}(y) \\
   & + w_{5}\, (1-d_{\mathrm{text}}(x,\tilde{x})),
\end{aligned}
\vspace{-6pt}
\end{equation}
with nonnegative weights $\{w_{1}, w_{2}, w_{3}, w_{4}, w_{5}\}$. The last term measures the consistency between the original text $x$ and recovered text $\tilde{x}$

% \subsection{Overall Training Pipeline}\label{subsec:pipeline}

\subsection{Training: SFT to DPO}\label{subsec:sft_dpo}

The training pipeline has three stages: (i) seed data synthesis, (ii) supervised fine-tuning (SFT) for structured mapping/recovery,
and (iii) direct preference optimization (DPO) for multi-objective alignment.

\paragraph{Stage 1: Data synthesis.}
We first use GPT-5.2 to synthesize a small subset\footnote{These instances are available in the provided GitHub repository.} of training instances  $\mathcal{D}_{\mathrm{synth}}=\{(x \xrightarrow{\mathcal{M}} (\tau,D,y)),\; ((\tau,D,y)\xrightarrow{\mathcal{R}} \tilde{x})\}$. 
% For each $x$, GPT-5.2 produces a structured mapper package $(\tau,D,y)$ by selecting a target-domain theme $\tau$,
% constructing a dictionary $D$, and rewriting $x$ into a mapped text $y$ in the target domain.
% It then generates a recovery $\tilde{x}$ conditioned on $(\tau,D,y)$.

\paragraph{Stage 2: SFT.}
The goal is to train two instruction-following models: a private mapper $\mathcal{M}_\phi$ and a recoverer $\mathcal{R}_\psi$.
Both are initialized from the same base LLM and fine-tuned with parameter-efficient adapters.
The mapper learns to generate $(\tau,D,y)$ from $x$, while the recoverer learns to generate $\tilde{x}$ from $(\tau,D,y)$.
% Concretely, SFT minimizes the conditional language modeling objectives:
% \vspace{-6pt}
% \begin{align}
%   \mathcal{L}_{\mathrm{SFT}}^{\mathrm{map}}(\phi)
%   &:=
%   \mathbb{E}
%   \Big[-\log p_{\phi}(\tau,D,y \mid x)\Big],
%    \nonumber \\
%   \mathcal{L}_{\mathrm{SFT}}^{\mathrm{rec}}(\psi)
%   &:=
%   \mathbb{E}
%   \Big[-\log p_{\psi}(\tilde{x} \mid \tau,D,y)\Big].
%   \vspace{-6pt}
%   \nonumber
%   \end{align}
  
  % SFT distills the structured mapping and recovery behavior from the seed data and serves as the initialization for the
  % subsequent preference optimization.

% but it has two limitations: (i) to avoide overfitting the small seed dataset, the SFT is performed with limited steps, and thus the insturction-following capability of the SFT mapper to produce desired structured outputs is not reliable. (ii) SFT does not explicity optimize the trade-offs between multiple objectives, e.g., privacy and recovery. We therefore refine the mapper with DPO using preference pairs constructed by sampling and scoring.

\paragraph{Stage 3: DPO.}
SFT establishes a strong base, but it does not directly optimize the multi-objective reward. DPO is therefore employed to refine the mapper with preference pairs obtained by sampling and scoring. During preference construction, for each $x$, we sample $K$ mapped candidates $(\tau^{(k)},D^{(k)},y^{(k)}) \sim p_{\phi}(\cdot \mid x),\quad k=1,\ldots,K$, and obtain a recovery for each candidate as $\tilde{x}^{(k)} \sim p_{\psi}(\cdot \mid \tau^{(k)},D^{(k)},y^{(k)})$. The scalar reward is computed as
$
r^{(k)} := r(x, \tau^{(k)},D^{(k)},y^{(k)},\tilde{x}^{(k)})
$, and a preference pair $(k^+,k^-)$ is selected such that $r^{(k^+)} \ge r^{(k^-)} + \gamma$ for a margin $\gamma>0$. Let $o^+ := o^{(k^+)}$ and $o^- := o^{(k^-)}$ denote the preferred and dispreferred mapper outputs, and $p_{\phi_0}$ be the frozen reference model (typically the SFT checkpoint).
DPO optimizes the mapper parameters $\phi$ by increasing the relative likelihood of high-reward mapped outputs over low-reward ones.

\subsection{Workflow with LLMs as Third-Party API}\label{subsec:inference}
% Algorithm~\ref{alg:inference} summarizes the workflow of the proposed CROSS-MAP with the trained mapper and recoverer.

Figure~\ref{fig:pipeline} illustrates the workflow of the proposed CROSS-MAP with the trained mapper and recoverer. After training, the private mapper $\mathcal{M}_\phi$ and the private recoverer $\mathcal{R}_\psi$ are deployed locally
(e.g., on the user side), while the downstream LLM is accessed as a third-party API. Before the transmission, $\mathcal{M}_\phi$ first maps the original input $x$ to a mapped version $y$. The LLM then produces a response $y_{\mathrm{a}}$ and transmits it back to the recoverer. \rev{The private recoverer can subsequently recover the original-domain answer $\tilde{x}_{\mathrm{a}}$ by reasoning jointly over the dictionary and the LLM response $y_{\mathrm{a}}$. Therefore, users can interact with external APIs or proprietary services under the protection of the local mapper and recoverer via CROSS-MAP. It allows the external model to operate on an alternative semantic domain while keeping the original sensitive content local.} Algorithm~\ref{alg:inference} in Appendix~\ref{sec:appendix_algorithm} summarizes this workflow.
\vspace{-3pt}
\subsection{Adversarial Attacks}\label{subsec:attack}
\vspace{-4pt}

To evaluate the security of CROSS-MAP, we simulate a white-box attacking scenario where the attacker has access to the mapper $\mathcal{M}_\phi$ and even to a partially leaked dictionary. Assume the attacker observes a mapped output $y$ and attempts to recover the original input $x$ using a parametric attack policy $p_{\theta}(\cdot |  y)$.

% Given an original input $x$, the mapper releases a mapped output
%   $o := (\tau,D,y)\sim \mathcal{M}_{\phi}(x)$ to an untrusted environment.
%   The adversary observes the released content and attempts to recover the original text.
%   Assume an adaptive attacker that can query the released mapper as a black-box oracle and optimize its attack strategy according to the mapper's output.
%   Specifically, the attacker maintains a parametric attack policy $p_{\theta}(\cdot \mid y)$ that proposes candidate reconstructions
%   $\hat{x}$ conditioned on the observation $y$.

\vspace{-6pt}
\paragraph{Attack objective and reward.}
Given a candidate reconstruction $\hat{x} \sim p_{\theta}(\cdot| y)$, the attacker can evaluate how well it explains the observation by passing $\hat{x}$ through the mapper to obtain $\hat{o} := (\hat{\tau},\hat{D},\hat{y}) \sim \mathcal{M}_{\phi}(\hat{x})$. A successful $\hat{x}$ should be mapped to $\hat{y}$ that is close to the observed $y$.
The attack reward is therefore defined as the similarity between $\hat{y}$ and $y$:
$\hat{r}(\hat{x};y) = 1 - d_{\mathrm{text}}(\hat{y},y)$.
The attacker aims to find a reconstruction $\hat{x}$ that maximizes $\hat{r}(\hat{x};y)$.

\vspace{-8pt}
\paragraph{DPO-based attacker.}
An optimization-based attacker can align $p_\theta(\cdot | y)$ to the reward $\hat{r}(\hat{x};y)$ via DPO. For each observed $y$, multiple candidate reconstructions $\{\hat{x}^{(k)}\}_{k=1}^{K}$ are sampled from $p_{\theta}(\cdot| y)$, then mapped by $\mathcal{M}_{\phi}$ to obtain $\hat{y}^{(k)}$, and finally scored by $\hat{r}^{(k)}=1 - d_{\mathrm{text}}(\hat{y}^{(k)},y)$.
A preference pair $(\hat{x}^+,\hat{x}^-)$ is then constructed by selecting two candidates with a reward gap
$\hat{r}(\hat{x}^+;y)\ \ge\ \hat{r}(\hat{x}^-;y)+\gamma_{\mathrm{atk}}$ for a margin $\gamma_{\mathrm{atk}}>0$. A detailed theoretical analysis of the attacker in the presence of dictionary leakage is provided in Appendix~\ref{sec:appendix_attack}.
% Let $p_{\theta_0}$ be a frozen reference attack policy. One DPO update minimizes
% \vspace{-6pt}
% \begin{equation*}
%   \begin{aligned}
%   \mathcal{L}_{\mathrm{DPO}}^{\mathrm{atk}}(\theta)
%   &:= -\mathbb{E}\!\left[
%   \log \sigma\!\left(
%   \beta_{\mathrm{atk}}\left(
%   \log\tfrac{p_{\theta}(\hat{x}^+\mid y)}{p_{\theta_0}(\hat{x}^+\mid y)}
%   \right.\right.\right. \\
%   &\left.\left.\left. \qquad -
%   \log\tfrac{p_{\theta}(\hat{x}^-\mid y)}{p_{\theta_0}(\hat{x}^-\mid y)}
%   \right)\right)
%   \right],
%   \end{aligned}
%   \vspace{-6pt}
%   \end{equation*}
% where $\beta_{\mathrm{atk}}>0$ controls the sharpness of the preference. 

% =========================
% Experiments (TEMPLATE, revised)
% - no bold metric names
% - avoid first-person "We"
% - restructure sections: separate (i) main trade-off, (ii) de-anonymization (attribute inference),
%   (iii) span restoration (black/white-box), (iv) RL attack (separate from ablation), (v) ablations
% - add bridging text between section/subsection/paragraph (no empty transitions)
% =========================

\vspace{-3pt}
\section{Experiments}
\label{sec:experiments}
\vspace{-3pt}

CROSS-MAP is evaluated along two dimensions:
\rev{(i) utility on question answering (QA), summarization, and generation tasks; and}
(ii) security against de-anonymization, black-box span restoration, and optimization-based white-box attacks.
\rev{Experiments are conducted on datasets from multiple semantic domains, including medicine, finance, news, and daily life.} A full description of the experimental setup, including the datasets and models used, is provided in Appendix~\ref{sec:appendix_exp_setup}.

\vspace{-5pt}
\subsection{Utility and Privacy}
\label{subsec:main_tradeoff}

% This subsection reports the key comparison of privacy and utility performance across datasets. A good mapper and recoverer should improve privacy while preserving utility compared to baselines.

\begin{figure}[t]
  \centering
  \includegraphics[width=0.95\linewidth]{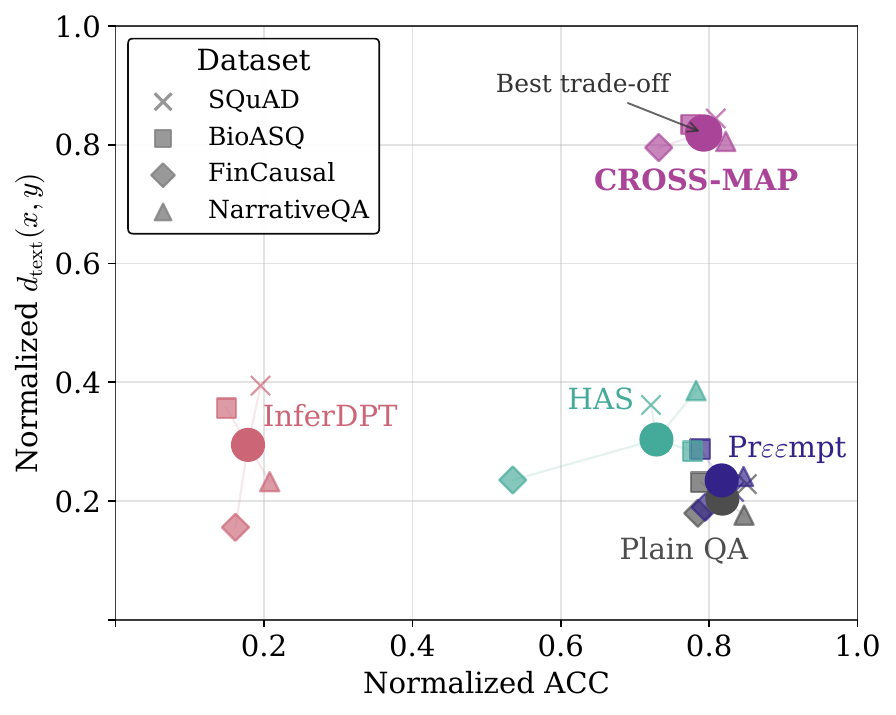}
  \vspace{-8pt}
  \caption{Privacy-utility trade-offs across methods and datasets. The round points denote the median of subpoints representing individual datasets.}
  \label{fig:pareto_tradeoff}
  \vspace{-12pt}
\end{figure}

% \paragraph{Trade-off visualization.}
Figure~\ref{fig:pareto_tradeoff} plots utility against privacy. It visualizes each method as a point in the plane of downstream QA accuracy and text distance $d_{\text{text}}(x, y)$ between original text $x$ and mapped text $y$ as described in Section~\ref{subsec:key_components}. Across datasets, Pr$\varepsilon$$\varepsilon$mpt~\citep{chowdhury2025pr} clusters near the upper-utility but low-privacy corner with Plain QA, indicating Pr$\varepsilon$$\varepsilon$mpt, which primarily targets numerical values, leaves the original text largely unchanged. InferDPT~\citep{tong2025inferdpt} typically shifts upward in privacy but drifts leftward in utility due to its lack of an explicit recovery process. HAS~\citep{chen2023hideseekhaslightweight} provides a better trade-off compared to InferDPT and Pr$\varepsilon$$\varepsilon$mpt. However, it is still inferior to CROSS-MAP, which consistently lies closest to the upper-right corner, achieving larger $d_{\text{text}}(x, y)$ than other baselines, and meanwhile maintains downstream QA performance very close to Plain QA. Table~\ref{tab:main_results_all} in Appendix~\ref{sec:appendix_exp_results} reports the concrete numbers. For example, on SQuAD, CROSS-MAP increases $d_\text{text}(x, y)$ from $0.238$ for HAS to $0.752$, while increasing QA accuracy from $91.84$ for HAS to $92.98$. Similar trends are observed on other datasets. \rev{We also report sensitive-span retention rate (the fraction of original sensitive spans retained in the mapped text) as a more direct privacy metric. Averaged across SQuAD, BioASQ, FinCausal, and NarrativeQA, CROSS-MAP retains only $3.1\%$ of sensitive spans, compared to $63.1\%$ for HAS and $100.0\%$ for Plain QA.} Table~\ref{tab:mapping_examples} in Appendix~\ref{sec:appendix_exp_results} provides examples of different mapping methods.

% \rev{Because embedding distances alone are indirect privacy proxies, we additionally measure sensitive-span retention: the fraction of original sensitive spans that survive in the mapped text. Averaged across SQuAD, BioASQ, FinCausal, and NarrativeQA, CROSS-MAP retains only $3.1\%$ of sensitive spans, compared with $63.1\%$ for HAS and $100.0\%$ for Plain QA. Across methods, datasets, and attack budgets, retention is strongly correlated with span-restoration ASR ($\rho=+0.926$, $p=1.7\times10^{-9}$), while $d_{\mathrm{text}}$ and $d_{\mathrm{dict}}$ are strongly negatively correlated with ASR ($\rho=-0.925$, $p=1.4\times10^{-14}$; $\rho=-0.823$, $p=4.1\times10^{-9}$). These results do not make embedding distance a complete privacy metric, but they validate that the reported distances track attack success in the evaluated threat models.}

\rev{To broaden the evaluation beyond QA, the SQuAD-trained mapper and recoverer are also evaluated on two non-QA task families without extra in-domain training. Table~\ref{tab:nonqa_main} summarizes the results. On CNN/DailyMail summarization, CROSS-MAP preserves summary quality, yielding a BERTScore of $0.877$ compared to $0.882$ for direct Plain LLM inference, while reducing the sensitive-span retention from $100.0\%$ to $7.3\%$. On CommonGen open-ended generation, CROSS-MAP preserves high concept coverage ($0.891$) while reducing retention from $98.3\%$ to $18.9\%$. These results support Theorem~\ref{thm:ib_sd_advantage}: semantic decoupling is strongest when the task-relevant structure is sufficient, but it can still provide a superior privacy-utility trade-off for semantic-dependent generation when the mapper performs domain-level dictionary construction and the recoverer performs soft semantic recovery instead of string replacement.}

\begin{table}[t]
  \centering
  \scriptsize
  \setlength{\tabcolsep}{2pt}
  \resizebox{\linewidth}{!}{%
  \begin{tabular}{llcccc}
    \toprule
    Task & Method & ROUGE-L & BERT/Cov. & Retention \\
    \midrule
    CNN/DM & Plain LLM  & 0.240 & 0.882 & 100.0\% \\
    CNN/DM & CROSS-MAP  & 0.285 & 0.877 & 7.3\% \\
    \midrule
    CommonGen & Plain LLM  & 0.587 & 0.932 & 98.3\% \\
    CommonGen & CROSS-MAP  & 0.414 & 0.891 & 18.9\% \\
    \bottomrule
  \end{tabular}
  }
  \caption{\rev{Generalization on summarization and generation tasks. The BERT/Cov. metric denotes BERTScore for summarization and concept coverage for generation. Retention represents the sensitive-span retention rate.}}
  \label{tab:nonqa_main}
  \vspace{-10pt}
\end{table}

% \begin{figure*}[t]
%   \centering
%   \begin{minipage}[t]{0.49\textwidth}
%     \centering
%     \includegraphics[width=\linewidth]{tradeoff.pdf}\\
%     \small\textbf{(a)} Aggregated privacy--utility scatter with Pareto frontier.
%   \end{minipage}\hfill
%   \begin{minipage}[t]{0.49\textwidth}
%     \centering
%     \includegraphics[width=\linewidth]{attr_infer_barchart.pdf}\\
%     \small\textbf{(b)} Privacy--utility trade-off per dataset, broken down by corpus.
%   \end{minipage}
%   \caption{Privacy--utility trade-off. X-axis: sentence-level privacy distance $d_{\text{sentence}}$ (higher indicates stronger privacy); Y-axis: downstream QA accuracy (higher indicates better utility).}
%   \label{fig:pareto_tradeoff2}
% \end{figure*}

% ==================================================
\subsection{Security Against Attacks}
This subsection reports and compares the security of different methods in three attack scenarios.

\subsubsection{De-anonymization Attacks}
\label{subsec:deanon_attr_infer}

LLM attribute inference is a strong way for de-anonymization. The goal is to use an LLM to infer private attributes  (e.g., gender, age, and location) from the mapped text. This subsection reports the success rate of LLM attribute inference on SynthPAI~\citep{yukhymenko2024synthetic} using GPT-4 as the attacker as in GPT-AA~\citep{staab2025large}. Figure~\ref{fig:attr_infer_barchart} shows that CROSS-MAP yields the lowest inference accuracy and certainty among all compared methods. The concrete numbers for each method are reported in Table~\ref{tab:attr_infer} in Appendix~\ref{sec:appendix_exp_results}.

\begin{figure}[t]
  \centering
  \includegraphics[width=\linewidth]{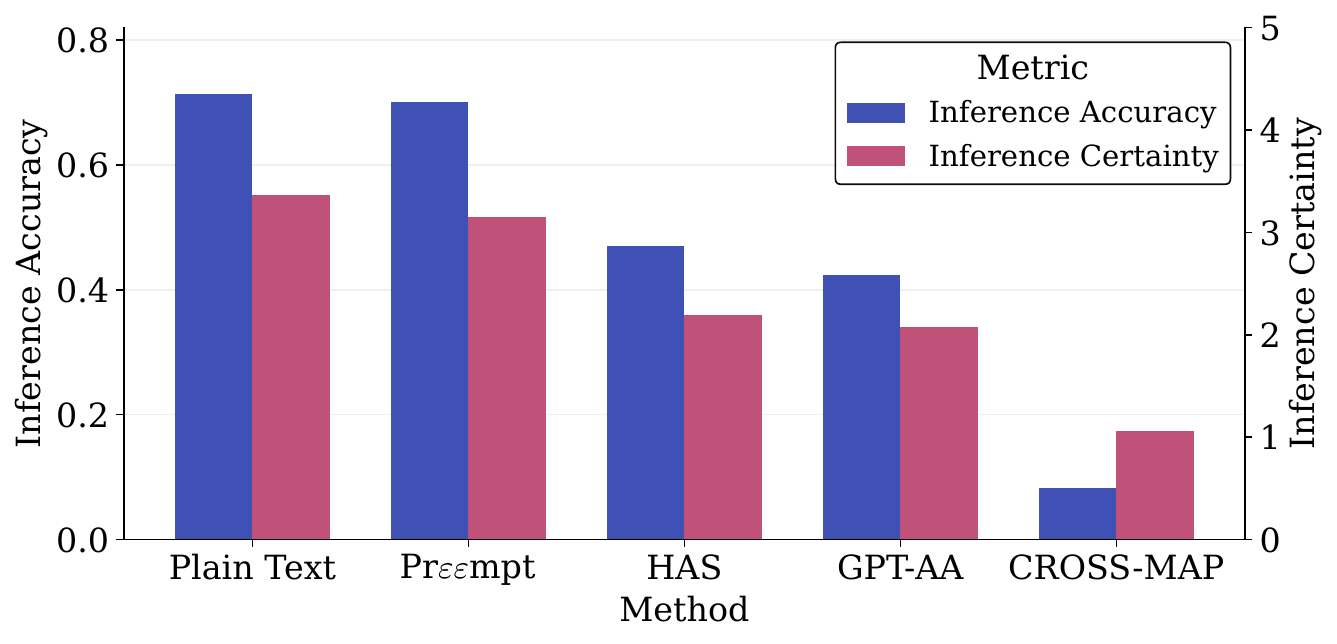}
  \vspace{-14pt}
  \caption{De-anonymization results via LLM attribute inference on SynthPAI~\citep{yukhymenko2024synthetic}, comparing attacker inference accuracy and certainty across methods.}
  \label{fig:attr_infer_barchart}
  \vspace{-12pt}
\end{figure}

% ==================================================

\begin{table*}[t]
\vspace{-10pt}
  \centering
  \scriptsize 
  \begin{tabular}{llcc|cc|cc}
    \toprule
    Dataset & Method
      & \multicolumn{2}{c|}{$k=10$}
      & \multicolumn{2}{c|}{$k=30$}
      & \multicolumn{2}{c}{$k=50$} \\
    \cmidrule(lr){3-4}\cmidrule(lr){5-6}\cmidrule(lr){7-8}
    & & ASR & \rev{Post. Mass} & ASR & \rev{Post. Mass} & ASR & \rev{Post. Mass} \\
    \midrule

    \multirow{3}{*}{FinCausal 2025~\cite{moreno2025financial}}
      & \textbf{\methodours} & \textbf{0.37$\%$} & \textbf{0.006} & \textbf{0.84$\%$} & \textbf{0.003} & \textbf{1.19$\%$} & \textbf{0.002} \\
      & \methodhas     & 5.88$\%$  & 0.116 & 35.29$\%$ & 0.026 & 35.29$\%$ & 0.026 \\
      & Plain QA       & 7.86$\%$  & 0.181 & 41.43$\%$ & 0.088 & 41.43$\%$ & 0.088 \\
    \midrule

    \multirow{3}{*}{BioASQ~\cite{tsatsaronis2015overview}}
      & \textbf{\methodours} & \textbf{1.05$\%$} & \textbf{0.004} & \textbf{1.05$\%$} & \textbf{0.004} & \textbf{1.05$\%$} & \textbf{0.004} \\
      & \methodhas     & 12.73$\%$ & 0.279 & 28.14$\%$ & 0.164 & 34.66$\%$ & 0.135 \\
      & Plain QA       & 16.29$\%$ & 0.316 & 32.81$\%$ & 0.239 & 36.12$\%$ & 0.208 \\
    \midrule

    \multirow{3}{*}{SQuAD~\cite{rajpurkar2016squad}}
      & \textbf{\methodours} & \textbf{0.68$\%$} & \textbf{0.068} & \textbf{1.22$\%$} & \textbf{0.042} & \textbf{1.49$\%$} & \textbf{0.023} \\
      & \methodhas     & 4.95$\%$  & 0.164 & 12.10$\%$ & 0.145 & 14.57$\%$ & 0.124 \\
      & Plain QA       & 5.33$\%$  & 0.209 & 15.70$\%$ & 0.171 & 15.70$\%$ & 0.171 \\
    \bottomrule
  \end{tabular}
  \vspace{-2pt}
  \caption{Black-box span restoration results across datasets, reporting span restoration accuracy as ASR and the attacker's posterior mass on the ground-truth span at different candidate list sizes $k$.}
  \label{tab:blackbox_span_restore_all}
  \vspace{-14pt}
\end{table*}

\vspace{-2pt}
\subsubsection{Black-box Span Restoration Attacks}
\label{subsec:span_restoration}
\vspace{-3pt}
Span restoration attacks aim to reconstruct the exact protected spans from the mapped text. The Attack Success Rate (ASR) is defined as the span restoration accuracy, calculated as the ratio of correctly restored spans to the total number of protected spans. Table~\ref{tab:blackbox_span_restore_all} shows that span restoration becomes easier as the candidate list size $k$ increases. Plain QA consistently yields the highest ASR, while \methodhas{} offers only limited improvement. In contrast, \methodours{} substantially reduces restoration success across all datasets and all values of $k$. This suggests that preserving semantic anchors makes posterior-based restoration easier for the attacker. \rev{The table also reports posterior mass on the ground-truth span, which is defined in Appendix~\ref{sec:appendix_asr}. Across datasets and different $k$, methods with lower posterior mass exhibit lower ASR. For instance, on FinCausal 2025 at $k=50$, \methodours{} assigns the lowest posterior mass ($0.002$) together with the lowest ASR ($1.19\%$), whereas \methodhas{} and Plain QA show higher posterior mass ($0.026$ and $0.088$) and correspondingly higher ASR ($35.29\%$ and $41.43\%$).} Figure~\ref{fig:blackbox_asr_superisal} in Appendix~\ref{sec:appendix_exp_results} visualizes such trend across datasets as $k$ varies. 

% Figure~\ref{fig:whitebox_asr} provides representative qualitative successes and failures that help explain why some methods are more vulnerable to span restoration attacks than others.

% ==================================================

\subsubsection{White-box Optimization-based Attacks}
\label{subsec:rl_attack}

The attacker is initialized from the trained mapper and then optimized using DPO algorithm as described in Section~\ref{subsec:attack}. ASR here refers to the similarity score $1-d_{\text{text}}(\hat{x},x)$ between the restored text $\hat{x}$ and the original text $x$. Table~\ref{tab:rl_attack_squad} and Table~\ref{tab:rl_attack_narrative} in Appendix~\ref{sec:appendix_exp_results} show that restoration becomes easier as the dictionary leakage ratio $\rho$ increases for both datasets, and that the DPO-based attacker consistently outperforms the base attacker without optimization at every $\rho$. Figure~\ref{fig:whitebox_asr} in Appendix~\ref{sec:appendix_exp_results} shows that higher leakage also helps DPO-based attacker to achieve larger gains over the base attacker. However, it is noticeable that even under this optimization-based white-box attacker, the improvement brought by optimization under CROSS-MAP remains bounded. Across both datasets and all leakage ratios $\rho$, the maximum relative gain of the DPO attacker over the base attacker is $29.4\%$, achieved on NarrativeQA at $\rho=0.1$ (where the DPO attacker's ASR is still only $24.2\%$). On SQuAD, the maximum relative gain is $13.6\%$ (achieved at $\rho=0.1$, with the DPO attacker's ASR still only $32.6\%$). This bounded improvement suggests that, even when the mapper's parameters are exposed and the attacker is allowed to optimize based on extra knowledge, e.g., some pairs from the mapping dictionary, the restoration advantage that can be extracted by optimization is still limited, supporting the robustness of CROSS-MAP against optimization-driven white-box restoration.

% \begin{figure}[t]
%   \centering
%   \includegraphics[width=\linewidth]{/Users/yuzhu/Projects/2025Fall/PhD-Project1/Overleaf/Association_for_Computational_Linguistics__ACL__conference/Gemini_Generated_Image_deb1dodeb1dodeb1.png}
%   \caption{Black-box span restoration examples, including sanitized answers, attacker reconstructions, and ground-truth spans that define success or failure.}
%   \label{fig:blackbox_examples}
% \end{figure}

% \paragraph{Effect of RL-based optimization.}
% On SQuAD (Table~\ref{tab:rl_attack}), the RL attacker consistently improves the mean similarity to the original text across leakage ratios $\rho$, with gains ranging from 0.039 at $\rho=0.1$ (0.287 to 0.326, a 13.6\% relative increase) to 0.094 at $\rho=0.9$ (0.773 to 0.867, a 12.2\% relative increase).
% The best-of-$M$ similarity also rises, from 0.397 to 0.935 as $\rho$ grows from 0.1 to 0.9, while the standard deviation remains below 0.046, indicating that RL yields not only higher but also relatively stable restoration quality.

% \paragraph{Analysis of reward proxy mismatch.}
% A key difficulty is the potential mismatch between the reward proxy and the true restoration objective:
% a near-correct guess at the original text may still yield large differences after forward mapping, which can distort optimization signals.

% ==================================================
\subsection{Ablation Study}
\label{subsec:ablation}
\rev{The ablation studies in this paper systematically evaluate the key factors impacting model deployment, including input complexity, paraphrase robustness, target-domain selection, dictionary coverage, training scale, external LLM capabilities, mapper/recoverer architecture, and training strategies. These findings offer empirical guidelines for optimal experimental configurations. Detailed tables and extended analyses are provided in Appendix~\ref{sec:appendix_exp_results}.}

\vspace{-6pt}
\section{Conclusion}
\vspace{-6pt}
This paper proposes CROSS-MAP, a semantic-decoupling framework for inference-time text protection. By replacing semantic content while preserving the structural patterns needed for downstream LLM reasoning, CROSS-MAP addresses the privacy–utility dilemma inherent in semantics-preserving approaches. An information-theoretic analysis characterizes how and when semantic decoupling provides a better privacy–utility trade-off. Experiments on QA, summarization, and generation tasks validate the analysis and demonstrate that CROSS-MAP outperforms existing baselines in both privacy and utility. These findings highlight the potential of semantic decoupling as a practical direction for privacy-preserving LLM inference.

\section{Limitations}
The utility of CROSS-MAP depends on the extent to which a downstream task can be supported by preserved structure, discourse roles, and recoverable relational information. The summarization and generation results show that the method is not restricted to extractive QA. \rev{However, the utility drop on CommonGen also confirms that highly semantic-dependent tasks remain more challenging}. One promising direction for future work is to augment CROSS-MAP with external relational knowledge sources, such as ConceptNet or related knowledge graphs, to better preserve entity-level relationships during mapping. This may be especially useful for downstream tasks that depend not only on abstract structural patterns but also on relational coherence among entities. \rev{Extending this framework to complex real-world tasks, including contract review, sensitive email drafting, and multi-turn reasoning, represents a critical next step.}

In addition, CROSS-MAP does not guarantee that mapped text will remain consistent with real-world facts or commonsense knowledge. Like other model-based text obfuscation methods, it may generate content that is factually implausible or inconsistent with external knowledge. Such artifacts can make obfuscated text easier to detect and may create opportunities for fact-based inference attacks, in which attackers use real-world constraints to narrow down the possible original meanings. Improving factual consistency under model-based text obfuscation remains an important direction for future work.

A further limitation is the potential dual-use risk of semantic decoupling. By replacing sensitive semantics while preserving structural patterns, the method could also be used to obscure malicious intent from external oversight while retaining enough structure for downstream reasoning. This creates a potential safety risk: harmful requests or plans may become less interpretable to LLM service providers even when their functional structure is preserved. Although CROSS-MAP is intended for privacy protection, its practical deployment should therefore incorporate safeguards against misuse.

% We proposed recoverable topic shifting for inference-time prompt privacy with an \SFT-to-\DPO training pipeline and instance-level inversion evaluation. Our approach provides a practical middle ground between irrecoverable anonymization and expensive cryptographic secure inference, while connecting privacy-preserving usage to semantic coding and adversarial evaluation.

% Use ref.bib for all citations (anthology.bib is not in this project; use \bibliography{custom,ref} if you add custom.bib entries)
\bibliography{ref}

% \clearpage

\appendix

\section{Appendix A: Full Literature Review}
\label{sec:appendix_related_work}

This section provides a comprehensive review of related work.

\subsection{Cryptography}
This line of work studies text encryption for privacy protection~\citep{zhang2025practical,wu2024ditto,limpcformer}. Many systems for agent-to-agent communication also provide support for standard encryption primitives~\citep{yuan2023secure}. While these techniques offer lossless protection that allows complete recovery, they induce nontrivial computational overhead and additional communication latency. Even when optimizations are applied~\citep{chen2024framework,hou2026ciphergpt,yubeaton2024truncformer}, the cost and engineering complexity remain substantial in interactive inference pipelines, particularly in long-context and latency-sensitive settings, which limits the practicality of cryptography as a general-purpose solution for inference-time text protection.

\subsection{Differential Privacy}

Differential Privacy (DP) achieves privacy protection by perturbing data-dependent computations with randomness. As LLMs progress, there are growing concerns about their capacity to memorize sensitive information from their input data. Existing approaches include applying DP during model training, such as using DP variants of stochastic optimization when fine-tuning on private datasets~\citep{abadi2016deep,wuprivately}. Another line of DP research focuses on training private generative models locally under DP constraints to synthesize data, which is then utilized for downstream fine-tuning or inference with larger models~\citep{yue2023synthetic,wang2025rewardds,wuprivacy,utpala2023locally,flemings2024differentially,hongdp}. These DP-based approaches share the core principle of injecting noise into data, gradients, or generation processes to avoid directly exposing the original private dataset to an external model. However, DP noise can significantly degrade model performance by perturbing information essential for reasoning and prediction. To preserve downstream utility, existing methods often either restrict DP protection to less critical components~\citep{chowdhury2025pr,tong2025inferdpt}, or focus on tasks that are relatively insensitive to exact factual and entity-level details, such as sentiment analysis~\citep{mattern2022differentially,kurakin2023harnessing}.

\subsection{Text Anonymization and Sanitization}
Text anonymization aims to transform text to reduce the risk of revealing personal information (e.g., identity or sensitive attributes) while preserving downstream utility~\citep{pilan2022text}. Early neural rewriting approaches combined feature-guided edits, such as dictionary-based synonym substitution, with sequence-to-sequence generation for anonymization~\citep{romanov2019natural,shetty2018a4nt}. Prior work demonstrates that modern LLMs have the potential to both anonymize and deanonymize text, highlighting risks of malicious re-identification attacks~\citep{patsakis2023man,staabbeyond}. Following research further proposes an LLM-based adversarial anonymization framework against LLM-driven re-identification~\citep{staab2025large}. Related attempts include multi-objective optimization for utility-preserving anonymization~\citep{yang2025robust}, misleading adversaries into predicting incorrect private attributes~\citep{frikha2025incognitext}, and span-level truthful anonymization via semantic generalization~\citep{pilan2025truthful}.

Compared to text anonymization, text sanitization extends beyond protecting personal identity information to broader sensitive content~\citep{yue2021differential,li2025papillon,chen2023customized}. To further improve text utility for downstream tasks, some studies introduce an explicit desanitization step. Existing methods include approaches based on predefined plaintext–ciphertext mappings~\citep{kan2023protectinguserprivacyremote}, methods that use a language model to generate span replacements and internalize transformation rules in a local model~\citep{shen2024fire,chen2023hideseekhaslightweight}, and designs that protect only specific sensitive types, e.g. numerical values, to preserve most of the original semantics~\citep{chowdhury2025pr}. This line of works share a consensus that the mapped text and original text must be within the same semantic space for more accurate LLM inference.

\section{Appendix B: Proof of Theorem \ref{thm:ib_sd_advantage}}
\label{sec:proof_ib_sd_advantage}

The formal proof of Theorem~\ref{thm:ib_sd_advantage} begins with the following definitions.

\paragraph{Preservation ratios $\rho_s$ and $\rho_c$.}
These ratios measure the amount of semantic and structural information from the original text that is preserved in the mapped text $y$:
\[
\rho_s=\frac{I(y_s;x_s)}{H(x_s)}, \qquad
\rho_c=\frac{I(y_c;x_c)}{H(x_c)},
\]
where $I(y_s;x_s)$ and $I(y_c;x_c)$ denote the mutual information between the semantic and structural components of the mapped text and those of the original text. $H(x_s)$ and $H(x_c)$ denote the entropy of the semantic and structural components of the original text.

\paragraph{Task-relevant preservation ratios $\eta_s$ and $\eta_c$.}
These ratios quantify how much task-relevant semantic and structural information is transferred from the original text to the mapped text:
\[
\eta_s=\frac{I(y_s;z)}{I(x_s;z)}, \qquad
\eta_c=\frac{I(y_c;z)}{I(x_c;z)},
\]
where $I(y_s;z)$ and $I(y_c;z)$ denote the mutual information between the semantic and structural components of the mapped text and the downstream task output $z$. $I(x_s;z)$ and $I(x_c;z)$ denote the mutual information between the semantic and structural components of the original text and the downstream task. These ratios characterize the fraction of task-relevant information preserved in the mapped text.

\paragraph{Task dependence coefficients $\lambda_s$ and $\lambda_c$.}
These coefficients describe the relative importance of semantic and structural information in the original text for the downstream task:
\[
\lambda_s=\frac{I(x_s;z)}{I(x;z)}, \qquad
\lambda_c=\frac{I(x_c;z)}{I(x;z)},
\]
where $I(x_s;z)$ and $I(x_c;z)$ denote the mutual information between the semantic and structural components of the original text and the downstream task output $z$, and $I(x;z)$ denotes the mutual information between the entire input text and the task output. Since
\[
I(x;z) = I(x_s;z) + I(x_c;z),
\]
it follows that
\[
\lambda_s + \lambda_c = 1.
\]

\paragraph{Compression gains $\Delta\rho_s$ and $\Delta\rho_c$.}
These quantities measure the difference in preserved information between the semantic-preserving and semantic-decoupling mappings:
\[
\Delta\rho_s=\rho_s^{SP}-\rho_s^{SD},\quad
\Delta\rho_c=\rho_c^{SP}-\rho_c^{SD},
\]
where $\rho_s^{SP}$ and $\rho_c^{SP}$ denote the preservation ratios under the semantic-preserving mapping, and $\rho_s^{SD}$ and $\rho_c^{SD}$ denote the preservation ratios under the semantic-decoupling mapping.

\paragraph{Task-relevant preservation differences $\Delta\eta_s$ and $\Delta\eta_c$.}
These quantities measure the difference in task-relevant information retained by the two mappings:
\[
\Delta\eta_s=\eta_s^{SP}-\eta_s^{SD},\quad
\Delta\eta_c=\eta_c^{SP}-\eta_c^{SD},
\]
where $\eta_s^{SP}$ and $\eta_c^{SP}$ denote the task-relevant preservation ratios under the semantic-preserving mapping, while $\eta_s^{SD}$ and $\eta_c^{SD}$ denote those under the semantic-decoupling mapping.

\paragraph{Proof of Theorem \ref{thm:ib_sd_advantage}.}

Consider the information bottleneck objective of a mapping $\mathcal{M}$:
\[
\mathcal{G}(\mathcal{M}) = I(x;y) - \beta I(y;z),
\]
where $x$ denotes the original text, $y$ the mapped text, and $z$ the downstream task variable.

Let the difference between the objectives of the semantic-decoupling mapping and the semantic-preserving mapping be
\begin{equation*}
\Delta = \mathcal{G}(\mathcal{M}_{\mathrm{SD}}) - \mathcal{G}(\mathcal{M}_{\mathrm{SP}}).
\end{equation*}

Expanding the definition yields
\begin{align*}
\Delta
&=
\big(I^{\mathrm{SD}}(x;y) - I^{\mathrm{SP}}(x;y)\big)
\\
&\quad
- \beta
\big(I^{\mathrm{SD}}(y;z) - I^{\mathrm{SP}}(y;z)\big).
\end{align*}

Define the compression gain
\[
\Delta C = I^{\mathrm{SP}}(x;y) - I^{\mathrm{SD}}(x;y),
\]
and the task-information difference
\[
\Delta T = I^{\mathrm{SP}}(y;z) - I^{\mathrm{SD}}(y;z).
\]

Then
\[
\Delta = -\Delta C + \beta \Delta T.
\]

Therefore,
\[
\mathcal{G}(\mathcal{M}_{\mathrm{SD}}) <
\mathcal{G}(\mathcal{M}_{\mathrm{SP}})
\]
holds whenever
\[
\Delta C > \beta \Delta T .
\]

\paragraph{Compression term.}

Assume that the semantic and structural components of the text are independent,
\[
x=(x_s,x_c), \qquad x_s \perp x_c .
\]

Under this assumption,
\[
I(x;y)=I(x_s;y_s)+I(x_c;y_c).
\]

Using the definitions of preservation ratios,
\[
I(x;y)=\rho_s H(x_s)+\rho_c H(x_c).
\]

The compression term is therefore
\begin{align*}
\Delta C
=
(\rho_s^{SP}-\rho_s^{SD})H(x_s)
\\+
(\rho_c^{SP}-\rho_c^{SD})H(x_c).
\end{align*}

Using the definitions
\[
\Delta\rho_s=\rho_s^{SP}-\rho_s^{SD},\,
\Delta\rho_c=\rho_c^{SP}-\rho_c^{SD},
\]
this becomes
\[
\Delta C
=
\Delta\rho_s H(x_s)+\Delta\rho_c H(x_c).
\]

Under the independence assumption, there is 
\[
H(x)=H(x_s)+H(x_c).
\]

Define the semantic information ratio
\[
r=\frac{H(x_s)}{H(x)}.
\]

Then
\[
\frac{H(x_s)}{H(x)}=r, \qquad
\frac{H(x_c)}{H(x)}=1-r.
\]

Thus
\[
\Delta C
=
H(x)\big(r\Delta\rho_s+(1-r)\Delta\rho_c\big).
\]

\paragraph{Task-information term.}

The mutual information between the mapped text and the downstream task can be decomposed as
\[
I(y;z)=I(y_s;z)+I(y_c;z).
\]

From the definitions of the task-relevant preservation ratios,
\[
I(y;z)
=
\eta_s I(x_s;z)
+
\eta_c I(x_c;z).
\]

Using the task dependence coefficients
\[
\lambda_s=\frac{I(x_s;z)}{I(x;z)}, \qquad
\lambda_c=\frac{I(x_c;z)}{I(x;z)},
\]
there is
\[
I(y;z)
=
(\eta_s\lambda_s+\eta_c\lambda_c)I(x;z).
\]

Therefore
\begin{align*}
\Delta T
&= (\eta_s^{SP}\lambda_s+\eta_c^{SP}\lambda_c)I(x;z) \\
&- (\eta_s^{SD}\lambda_s+\eta_c^{SD}\lambda_c)I(x;z)\\
&=(\Delta\eta_s\lambda_s+\Delta\eta_c\lambda_c)I(x;z),
\end{align*}
where
\[
\Delta\eta_s=\eta_s^{SP}-\eta_s^{SD},\qquad
\Delta\eta_c=\eta_c^{SP}-\eta_c^{SD}.
\]

\paragraph{Combining the two terms.}

The condition
\[
\Delta C > \beta \Delta T,
\]
becomes
\begin{align*}
&H(x)\big(r\Delta\rho_s+(1-r)\Delta\rho_c\big)
\\&>
\beta(\Delta\eta_s\lambda_s+\Delta\eta_c\lambda_c)I(x;z).
\end{align*}

Let
\[
\beta'=\beta \cdot \frac{I(x;z)}{H(x)}.
\]

Dividing both sides by $H(x)$ gives
\[
r\Delta\rho_s+(1-r)\Delta\rho_c
>
\beta'(\Delta\eta_s\lambda_s+\Delta\eta_c\lambda_c).
\]

Since
\[
\lambda_c=1-\lambda_s,
\]
the right-hand side becomes
\[
\beta'(\Delta\eta_s\lambda_s+\Delta\eta_c(1-\lambda_s)).
\]

Solving the inequality for $\lambda_s$ yields
\[
\lambda_s <
\frac{
r \Delta\rho_s + (1-r)\Delta\rho_c - \beta' \Delta\eta_c
}{
\beta'(\Delta\eta_s-\Delta\eta_c)
}.
\]

Whenever this condition holds, the inequality
\[
\mathcal{G}(\mathcal{M}_{\mathrm{SD}})
<
\mathcal{G}(\mathcal{M}_{\mathrm{SP}}),
\]
is satisfied. This completes the proof.

% \paragraph{Interpretation.}
% The left-hand side represents the compression gain obtained by
% removing redundant semantic information, while the right-hand side
% measures the loss of task-relevant information. Semantic-decoupling
% mappings are therefore advantageous when the downstream task relies
% more on structural information than on lexical semantics.

\section{Appendix C: Analysis of Attack Success Rate}
\label{sec:appendix_asr}
This section supplements the analysis of the relationship between KL divergence and attack success rate in Section~\ref{sec:theory}.

\begin{figure*}[t]
  \centering
  \includegraphics[width=0.9\textwidth]{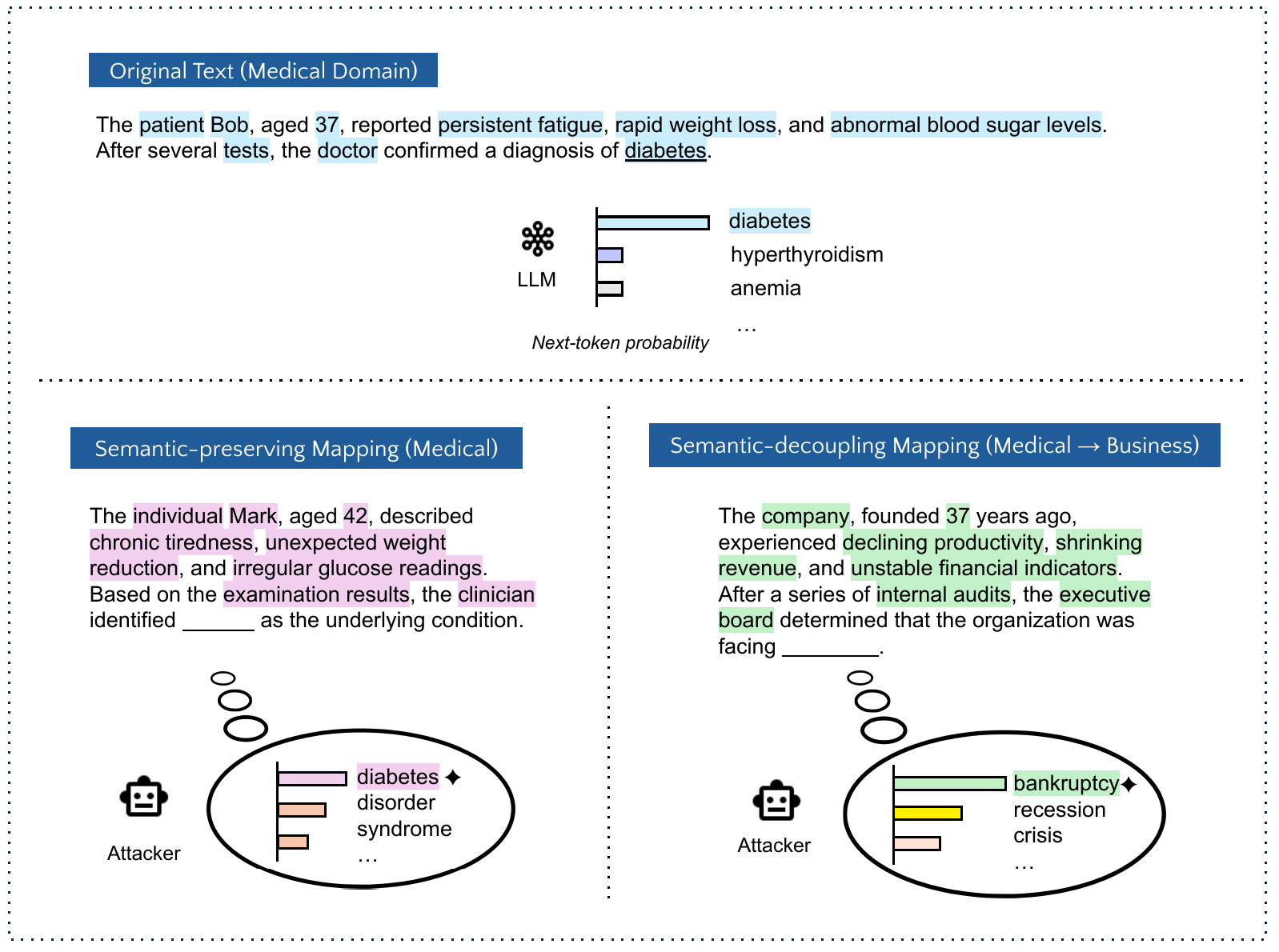}
  \caption{Examples of semantic-preserving and semantic-decoupling mappings under posterior sampling attacks. Semantic decoupling prevents correct inferences by stripping contextual cues from the source domain.}
  \label{fig:attack_examples}
\end{figure*}

\paragraph{Posterior-based span inference attack.}
Let $x$ be the original text, and let $s^\star := (s_1^\star,\ldots,s_m^\star)$ denote the $m$ sensitive spans in $x$ protected by the mapping mechanism.
Given the attacker-observed output $y=\mathcal{M}(x)$, the attacker aims to recover each hidden span $s_i^\star$ from $y$.

Let $q_\theta(s_i| y)$ denote the attacker model, e.g., an LLM posterior over the $i$-th span conditioned on $y$. For each span $i\in\{1,\ldots,m\}$, the attacker outputs $K$ guesses
\[
\hat{s}_i^{(1)},\ldots,\hat{s}_i^{(K)} \sim q_\theta(s_i| y).
\]

Equivalently, the attacker forms a candidate set
\[
\widehat{\mathcal{C}}_i^{K}(y)
:=
\{\hat{s}_i^{(1)},\ldots,\hat{s}_i^{(K)}\}.
\]

A span is said to be successfully recovered if the ground-truth span appears in the candidate set, i.e.,
\[
s_i^\star \in \widehat{\mathcal{C}}_i^{K}(y).
\]

\paragraph{Attack success rate at top-$K$.}
We define the per-example attack success rate at top-$K$ as the fraction of sensitive spans successfully recovered:
\[
\mathrm{ASR@}K(x,y)
:=
\frac{1}{m}\sum_{i=1}^m
\mathbf{1}\!\left\{s_i^\star \in \widehat{\mathcal{C}}_i^{K}(y)\right\}.
\]

Taking expectation over the data distribution and mapping randomness gives

\vspace{-6pt}
{\small
\[
\mathrm{ASR@}K(\mathcal{M})
:=
\mathbb{E}_{x\sim p(x),\, y\sim \mathcal{M}(x)}
\left[\mathrm{ASR@}K(x,y)\right].
\]
}

Thus, $\mathrm{ASR@}K(\mathcal{M})$ measures the expected proportion of sensitive spans that can be recovered by the attacker using $K$ guesses per span.

\vspace{6pt}
\begin{proposition}[$\mathrm{ASR@}K$ under posterior sampling]
  \label{prop:ask_posterior}
  Assume that for each sensitive span $i$, the attacker draws $K$ i.i.d. guesses from the posterior $q_\theta(s_i| y)$. Then
  
  \vspace{-8pt}
  {\small
  \[
  \mathrm{ASR@}K(\mathcal{M})
  =
  \mathbb{E}_{x,y}\left[
  \frac{1}{m}\sum_{i=1}^m
  \Bigl(1-\bigl(1-q_\theta(s_i^\star\mid y)\bigr)^K\Bigr)
  \right].
  \]
  }
  \end{proposition}
  
  \begin{proof}
  For a fixed $(x,y)$ and a fixed span $i$, let
  \[
  p_i(y):=q_\theta(s_i^\star\mid y).
  \]

  Since the attacker draws $K$ i.i.d. guesses from $q_\theta(s_i| y)$, the probability that none of the $K$ guesses equals the true span $s_i^\star$ is
  \[
  (1-p_i(y))^K.
  \]

  Therefore, the probability that the true span is recovered at least once is
  \[
  1-(1-p_i(y))^K
  =
  1-\bigl(1-q_\theta(s_i^\star\mid y)\bigr)^K.
  \]

  Taking the average over the $m$ spans gives

  \vspace{-6pt}
  {\small
  \[
  \mathbb{E}\!\left[\mathrm{ASR@}K(x,y)\mid x,y\right]
  =
  \frac{1}{m}\sum_{i=1}^m
  \Bigl(1-\bigl(1-q_\theta(s_i^\star\mid y)\bigr)^K\Bigr).
  \]
  }

  Finally, taking expectation over $(x,y)$ yields the result.
  \end{proof}

  \paragraph{Connection to KL divergence.}
  For each sensitive span $s_i^\star$, define its surprisal~\cite{kakouros2023investigating} under the attacker posterior as
  \[
  \mathcal{S}_i(y)
  :=
  -\log_2 q_\theta(s_i^\star\mid y).
  \]

  Taking expectation over the data distribution
\begin{align*}
  \mathbb{E}\bigl[\mathcal{S}_i(\mathcal{Y})\bigr]
  &=
  H(\mathcal{S}_i(\mathcal{Y})\mid \mathcal{Y}) \\
  &=
  H(\mathcal{S}_i(\mathcal{Y}))-I(\mathcal{S}_i(\mathcal{Y});\mathcal{Y}),
  \end{align*}
  where
  \[
  I(\mathcal{S}_i(\mathcal{Y});\mathcal{Y})
  =
  \mathbb{E}_{y}\!\left[\mathrm{KL}\!\bigl(p(s_i\mid y)\,\|\,p(s_i)\bigr)\right].
  \]

  Therefore, smaller KL divergence implies larger expected surprisal. Furthermore, since
  \[
  q_\theta(s_i^\star\mid y)=2^{-\mathcal{S}_i(y)},
  \]
  then Proposition~\ref{prop:ask_posterior} can be rewritten as

  \vspace{-8pt}
  {\small
  \[
  \mathrm{ASR@}K(\mathcal{M})
  =
  \mathbb{E}_{x,y}\left[
  \frac{1}{m}\sum_{i=1}^m
  \Bigl(1-\bigl(1-2^{-\mathcal{S}_i(y)}\bigr)^K\Bigr)
  \right].
  \]
  }

Therefore, $\mathrm{ASR@}K$ is a monotonically decreasing function of the span surprisal. Lower surprisal assigns larger posterior mass to the ground-truth span, which increases the probability that the true span appears among the attacker's $K$ guesses. Consequently, a larger posterior-to-prior KL divergence leads to smaller expected surprisal and hence higher $\mathrm{ASR@}K$.

This aligns with the qualitative conclusion of this section: semantic-preserving mappings tend to be more vulnerable because they preserve cues in $y$ that concentrate the posterior around $s^\star$. Figure~\ref{fig:attack_examples} provides an example of such pattern.

\section{Appendix D: Algorithm}
\label{sec:appendix_algorithm}
The details of the inference workflow with CROSS-MAP described in Section~\ref{subsec:inference} are summarized in Algorithm~\ref{alg:inference}.

\begin{algorithm}[t]
  \caption{Cross-domain Mapping and Recovery}
  \label{alg:inference}
  \begin{algorithmic}[1]
  \REQUIRE Trained mapper $\mathcal{M}_\phi$; trained recoverer $\mathcal{R}_\psi$; downstream $\mathcal{LLM}$.
  \FOR{each input $x$}
    \STATE $(\tau, D, y)\sim \mathcal{M}_\phi(x)$.
    \STATE $y_{\mathrm{a}} \leftarrow \mathcal{LLM}(y)$.
    \STATE $(\tilde{x}, \tilde{x}_{\mathrm{a}}) \sim \mathcal{R}_\psi\!\big(\tau, D, \mathrm{concat}(y, y_{\mathrm{a}})\big)$.
    \STATE \textbf{return} $\tilde{x}_{\mathrm{a}}$.
  \ENDFOR
  \end{algorithmic}
  \end{algorithm}

\section{Appendix E: Security Analysis of CROSS-MAP}
  \label{sec:appendix_attack}
  
  This section analyzes the security of CROSS-MAP against posterior-based attackers under dictionary leakage, when the attacker additionally observes or infers a subset of the mapper dictionary.
  
  As discussed in Appendix~\ref{sec:appendix_asr}, the security of a mapping mechanism can be characterized through the posterior mass assigned to the ground-truth sensitive spans. In particular, under posterior-based attacks with $K$ guesses per span, the attack success rate $\mathrm{ASR@}K$ is determined by the posterior probabilities of the true spans, equivalently by their surprisal. We now analyze how these quantities change when part of the dictionary is leaked to the attacker.
  
  \paragraph{Leakage model.}
  Let $  s^\star=(s_1^\star,\ldots,s_m^\star)$ denote the $m$ sensitive spans in the original text $x$, and let $\hat{s}=(\hat{s}_1,\ldots,\hat{s}_m)$ be their mapped-domain counterparts produced by CROSS-MAP. The full mapper dictionary for this example is
  \[
  D=\{(s_i^\star,\hat{s}_i)\}_{i=1}^{m}.
  \]

  Assume the attacker additionally learns a leaked subset $D_{\mathrm{leak}}\subseteq D$ containing $\ell$ revealed pairs, with leakage rate $\rho=\ell/m$. Let $\mathcal{I}_{\mathrm{leak}}\subseteq\{1,\ldots,m\}$ denote the indices of leaked spans, and let $\mathcal{I}_{\mathrm{hid}}:=\{1,\ldots,m\}\setminus \mathcal{I}_{\mathrm{leak}}$ be the indices of the still-hidden spans. Given the attacker-observed mapped output $y=\mathcal{M}(x)$, dictionary leakage enriches the attacker-visible information from $y$ to $(y,D_{\mathrm{leak}})$.

  Accordingly, for each hidden span $i\in\mathcal{I}_{\mathrm{hid}}$, the attacker posterior changes from $q_\theta(s_i\mid y)$ to the leakage-conditioned posterior $q_\theta^\rho(s_i\mid y,D_{\mathrm{leak}})$. For leaked spans $i\in\mathcal{I}_{\mathrm{leak}}$, recovery is trivial once the corresponding dictionary pair is known.
  
  \paragraph{Attack success rate under leakage.}
  Under the same attack protocol as in Appendix~\ref{sec:appendix_asr}, the attacker outputs $K$ guesses for each span. For leaked spans, the correct value is directly recovered from $D_{\mathrm{leak}}$. For hidden spans, the attacker draws $K$ guesses from the leakage-conditioned posterior $q_\theta^\rho(s_i | y,D_{\mathrm{leak}})$. Define the per-example success rate under leakage as

  \vspace{-12pt}
  {\scriptsize
  \[
  \mathrm{ASR@}K_\rho(x,y)
  :=
  \frac{1}{m}
  \left(
  \sum_{i\in\mathcal{I}_{\mathrm{leak}}} 1
  +
  \sum_{i\in\mathcal{I}_{\mathrm{hid}}}
  \mathbf{1}\!\left\{s_i^\star\in \widehat{\mathcal S}_{i,\rho}^K(y)\right\}
  \right),
  \]}
where $\widehat{\mathcal S}_{i,\rho}^K(y)$ denotes the attacker's $K$ guesses for span $i$ under leakage. Taking expectation over the data distribution and mapping randomness gives
  \[
  \mathrm{ASR@}K_\rho(\mathcal M)
  :=
  \mathbb E_{x,y}\!\left[\mathrm{ASR@}K_\rho(x,y)\right].
  \]
  
  \begin{proposition}[$\mathrm{ASR@}K_\rho$ under dictionary leakage]
  \label{prop:asrk_leakage}
  Assume that for each hidden span $i\in\mathcal{I}_{\mathrm{hid}}$, the attacker draws $K$ i.i.d. guesses from
  $q_\theta^\rho(s_i| y,D_{\mathrm{leak}})$.
  Then

  \vspace{-8pt}
  {\small
  \begin{align*}
  &\mathrm{ASR@}K_\rho(\mathcal M)
  =
  \mathbb E_{x,y}\!\left[
  \frac{1}{m}
  \left(
  |\mathcal{I}_{\mathrm{leak}}|
  \right.
  \right.
  \\
  &\quad
  \left.
  \left.
  +
  \sum_{i\in\mathcal{I}_{\mathrm{hid}}}
  \Bigl(1-\bigl(1-q_\theta^\rho(s_i^\star\mid y,D_{\mathrm{leak}})\bigr)^K\Bigr)
  \right)
  \right].
  \end{align*}
  }
  \end{proposition}
  
  \begin{proof}
  For each leaked span $i\in\mathcal{I}_{\mathrm{leak}}$, the true span is directly revealed by the leaked dictionary pair, so its contribution to the success rate is $1$.
  
  For each hidden span $i\in\mathcal{I}_{\mathrm{hid}}$, let
  \[
  p_i^\rho(y):=q_\theta^\rho(s_i^\star\mid y,D_{\mathrm{leak}}).
  \]

  Since the attacker draws $K$ i.i.d. guesses from the leakage-conditioned posterior, the probability that none of the $K$ guesses equals the ground-truth span is
  \[
  (1-p_i^\rho(y))^K.
  \]

  Therefore, the probability that the true span is recovered at least once is
  \[
  1-(1-p_i^\rho(y))^K
  =
  1-\bigl(1-q_\theta^\rho(s_i^\star\mid y,D_{\mathrm{leak}})\bigr)^K.
  \]

  Averaging over all $m$ spans yields the conditional per-example success rate, and taking expectation over $(x,y)$ gives the result.
  \end{proof}
  
  \paragraph{Connection to surprisal under leakage.}
  For each hidden span $i\in\mathcal{I}_{\mathrm{hid}}$, define its leakage-conditioned surprisal as
  \[
  \mathcal S_i^\rho(y)
  :=
  -\log_2 q_\theta^\rho(s_i^\star\mid y,D_{\mathrm{leak}}).
  \]

  Then Proposition~\ref{prop:asrk_leakage} can be rewritten as
  {\small
  \begin{align*}
  \mathrm{ASR@}K_\rho(\mathcal M)
  &=
  \mathbb E_{x,y}\!\left[
  \frac{1}{m}
  \left(
  |\mathcal{I}_{\mathrm{leak}}|
  +\right.
  \right.
  \\
  &\quad
  \left.
  \left.
  \sum_{i\in\mathcal{I}_{\mathrm{hid}}}
  \Bigl(1-\bigl(1-2^{-\mathcal S_i^\rho(y)}\bigr)^K\Bigr)
  \right)
  \right].
  \end{align*}
  }

  Hence, for the unrevealed spans, dictionary leakage increases attack success exactly when it decreases the surprisal of the ground-truth spans under the attacker's posterior.
    
  \paragraph{Effect 1: Search-space reduction for hidden spans.}
  Dictionary leakage reduces the effective search space for the remaining hidden spans. Intuitively, once some span--mapping pairs are known, candidates inconsistent with the leaked dictionary can be eliminated, which increases the relative posterior mass assigned to the feasible values of the unrevealed spans.
  
  To make this precise, suppose that for each hidden span $i$, the attacker originally considers a feasible candidate set $\mathcal C_i(y)$, while under leakage the feasible set shrinks to
  \[
  \mathcal C_i^\rho(y,D_{\mathrm{leak}})\subseteq \mathcal C_i(y).
  \]

  If the posterior over the feasible set is approximately uniform, then
  \[
  q_\theta(s_i^\star\mid y)\approx \frac{1}{|\mathcal C_i(y)|},
  \]
  \[
    q_\theta^\rho(s_i^\star\mid y,D_{\mathrm{leak}})
    \approx \frac{1}{|\mathcal C_i^\rho(y,D_{\mathrm{leak}})|}.
  \]
  
  Since
  \[
  |\mathcal C_i^\rho(y,D_{\mathrm{leak}})|
  \le
  |\mathcal C_i(y)|,
  \]
  the posterior mass of the true span increases after leakage, which lowers surprisal and raises $\mathrm{ASR@}K$.
  
  \paragraph{Effect 2: Posterior concentration via cross-span dependence.}
  Beyond search-space reduction, leaked spans can provide additional semantic or structural clues about the remaining hidden spans. Let $s_{\mathrm{hid}}$ and $s_{\mathrm{leak}}$ denote the collections of hidden and leaked spans, respectively. Conditioning on the leaked dictionary cannot increase uncertainty:
  \[
  H(s_{\mathrm{hid}}\mid y,D_{\mathrm{leak}})
  \le
  H(s_{\mathrm{hid}}\mid y).
  \]

  Equivalently,
  \begin{align*}
  &I(s_{\mathrm{hid}};D_{\mathrm{leak}}\mid y) \\
  &=
  H(s_{\mathrm{hid}}\mid y)-H(s_{\mathrm{hid}}\mid y,D_{\mathrm{leak}})\ge 0.
  \end{align*}

  Thus, leaked dictionary entries provide additional information about the unrevealed spans beyond what is already contained in $y$. At the per-span level, for each hidden span $s_i$ this implies
  \[
  H(s_i\mid y,D_{\mathrm{leak}})
  \le
  H(s_i\mid y).
  \]

  If the attacker posterior is well calibrated to the true posterior, then the expected leakage-conditioned surprisal satisfies
  \begin{align*}
  \mathbb E[\mathcal S_i^\rho(\mathcal{Y})]
  =
  H(s_i\mid \mathcal{Y},D_{\mathrm{leak}}) \\
  \le
  H(s_i\mid \mathcal{Y})
  =
  \mathbb E[\mathcal S_i(\mathcal{Y})].
  \end{align*}

  Therefore, dictionary leakage reduces the expected surprisal of the unrevealed spans and increases their recovery probability under posterior-based attacks.
  
  % \begin{proposition}[Monotonicity of attack success under nested leakage]
  % \label{prop:monotone_leakage}
  % Assume a nested leakage mechanism: for any $\rho_1\le \rho_2$, the leaked dictionary at rate $\rho_2$ contains all pairs leaked at rate $\rho_1$. Assume further that for every hidden span $i$, the corresponding leakage-conditioned posterior mass on the true span is non-decreasing with leakage:
  % \[
  % q_\theta^{\rho_2}(s_i^\star\mid y,D_{\mathrm{leak}}^{\rho_2})
  % \ge
  % q_\theta^{\rho_1}(s_i^\star\mid y,D_{\mathrm{leak}}^{\rho_1}).
  % \]
  % Then
  % \[
  % \mathrm{ASR@}K_{\rho_2}(\mathcal M)
  % \ge
  % \mathrm{ASR@}K_{\rho_1}(\mathcal M).
  % \]
  % \end{proposition}
  
  % \begin{proof}
  % For leaked spans, increasing $\rho$ can only convert previously hidden spans into directly recovered spans, which increases their contribution to the average success rate.
  
  % For the remaining hidden spans, Proposition~\ref{prop:asrk_leakage} shows that the success term is
  % \[
  % 1-\bigl(1-q\bigr)^K,
  % \]
  % where $q$ is the posterior mass assigned to the true span. This function is monotone increasing in $q$ on $[0,1]$. Hence, if the posterior mass on the true span is non-decreasing under stronger leakage, then each hidden-span contribution is also non-decreasing. Averaging over spans and taking expectation proves the result.
  % \end{proof}
  
  The above analysis shows that dictionary leakage harms security through two coupled mechanisms: it directly reveals the leaked spans themselves, and it indirectly increases the recoverability of the remaining spans by concentrating the attacker posterior. In terms of the CROSS-MAP framework developed in this paper, both effects increase the posterior mass on the ground-truth spans, reduce their surprisal, and therefore raise $\mathrm{ASR@}K$. This motivates evaluating security as a function of the leakage rate $\rho$ and reporting robustness against strong attackers with partial dictionary knowledge inferred from historical observations.

\section{Appendix F: Supplementary Experimental Results}
This section provides supplementary details on the experimental setup and results.

\subsection{Experimental Setup}
\label{sec:appendix_exp_setup}

This subsection reports the experimental setup, including datasets, baselines, models, and evaluation metrics.

\paragraph{Datasets.} \rev{SQuAD~\cite{rajpurkar2016squad}, BioASQ~\cite{tsatsaronis2015overview}, FinCausal 2025~\cite{moreno2025financial}, and NarrativeQA~\cite{kovcisky2018narrativeqa} are used for QA-style evaluation across general, biomedical, financial, and narrative domains. CNN/DailyMail is used for summarization, and CommonGen is used for constrained open-ended generation.} SynthPAI~\citep{yukhymenko2024synthetic} is used for evaluation against de-anonymization attacks. Model optimization involves a subset of SQuAD comprising 754 instances for SFT and 1,502 instances for DPO. For QA evaluation, 500 test instances are sampled from each dataset. The additional summarization and open-ended generation experiments use 200 CNN/DailyMail instances and 100 CommonGen instances, respectively. All reported results are averaged across the corresponding test subset in a single run for each dataset.
% For each dataset, a controlled subset of $N=\todo{50}$ instances is selected following \todo{sampling rule, filters, and random seed}.

\paragraph{Baselines.} State-of-the-art methods for inference-time privacy protection are used as baselines, including the sanitization method InferDPT~\cite{tong2025inferdpt}, and bidirectional frameworks HAS~\cite{chen2023hideseekhaslightweight} and Pr$\varepsilon$$\varepsilon$mpt~\cite{chowdhury2025pr}. For evaluation against de-anonymization attacks, the adversarial anonymization framework GPT-AA~\cite{staab2025large} is also included as a baseline.

% \paragraph{LLM/API usage and prompts.}
% A single GPT model (\todo{model name/version}) is used to generate answers for each method on four datasets
% ($4\times 50$ samples), and to run the de-anonymization evaluation on PersonalReddit (up to $4\times 50$ samples).
% For black-box span restoration, a Gemini model (\todo{model name/version}) is used as the attacker.
% Decoding settings are fixed across datasets: temperature=\todo{--}, top-$p$=\todo{--}, max tokens=\todo{--}.
% Prompt templates are included in Appendix~\todo{X} to ensure reproducibility.

\paragraph{Metrics.}
\rev{The evaluation focuses on the trade-off between utility and privacy across tasks. For QA, utility is assessed with EM/F1/ACC. For summarization, utility is assessed with ROUGE-L and BERTScore. For CommonGen, utility is assessed with ROUGE-L and concept coverage. The utility-oriented training objectives introduced in Section~\ref{sec:method}, including fluency and recovery quality, are also reported where applicable. Privacy is measured using dictionary-/text-level distance, dictionary coverage, sensitive-span retention, and the posterior-mass quantity derived from surprisal in Appendix~\ref{sec:appendix_attack}. Beyond these intrinsic metrics, attack success rate (ASR) is reported under different attackers to reflect effective privacy under each threat model.}

\paragraph{Models.} Qwen-2.5-3B/7B/14B~\cite{qwen2025qwen25technicalreport} and Qwen3-4B~\cite{yang2025qwen3} are used as the base models for the mapper and the recoverer. LLAMA-8B is used as the downstream model for the main experiments, \rev{and the additional deployment ablation also evaluates stronger external LLM access through CROSS-MAP.} LoRA is used for parameter-efficient fine-tuning with $r=16$, $\alpha=32$, and a learning rate of $1\times 10^{-5}$ for all runs. 

\paragraph{Prompts.} 
Figure \ref{fig:mapper_prompt} and \ref{fig:recoverer_prompt} show the mapper and recoverer prompts, respectively.

\subsection{Supplementary Experimental Results}
This section provides a full description of experimental results, including tables and figures, as a supplement to Section~\ref{sec:experiments}.

\label{sec:appendix_exp_results}
\subsubsection{Utility and Privacy}

\begin{table*}[t]
  \centering
  \scriptsize
  \setlength{\tabcolsep}{2.6pt}
  \begin{tabular}{llccc|cccc|ccc}
    \toprule
    & & \multicolumn{3}{c|}{\textbf{Downstream QA utility}}
    & \multicolumn{4}{c|}{\textbf{Recovered text utility}}
    & \multicolumn{3}{c}{\textbf{Mapped text privacy}} \\
    \cmidrule(lr){3-5}\cmidrule(lr){6-9}\cmidrule(lr){10-12}
    Dataset & Method
    & {\rev{EM$\uparrow$}} & {\rev{F1$\uparrow$}} & {\rev{ACC$\uparrow$}}
    & {\rev{Flu. Loss$\downarrow$}} & {\rev{Sim.$\uparrow$}} & {\rev{BLEU$\uparrow$}} & {\rev{R-1$\uparrow$}}
    & {\rev{$c_\text{dict}\uparrow$}} & {\rev{$d_{\text{dict}}\uparrow$}} & {\rev{$d_{\text{text}}\uparrow$}}   \\
    \midrule
    \multirow{5}{*}{SQuAD}
      & Plain QA        & 76.42 & 84.10 & 94.73 & 4.202 & 1.000 & 1.000 & 1.000 & 0.000 & 0.000 & 0.000  \\
      & Pr$\varepsilon$$\varepsilon$mpt  & 75.78 & 84.02 & 94.70 & 4.365 & 1.000 & 1.000 & 1.000 & 0.006 & 0.142 & 0.006   \\\noalign{\vskip 0.3em}
      \cdashline{2-12} \noalign{\vskip 0.5em}
      & InferDPT   & 68.21 & 74.04 & 81.34 & 5.597 & 0.737 & 0.520 & 0.710 & 0.172 & 0.289 & 0.263  \\
      & HAS      & 72.97 & 82.10 & 91.84 & 4.212 & 0.973 & 0.939 & 0.975 & 0.181 & 0.408 & 0.238  \\
      & \textbf{CROSS-MAP}
                        & \textbf{74.94} & \textbf{83.22} & \textbf{92.98}
                        & 4.703 & 0.826 & 0.586 & 0.800
                        & \textbf{0.413} & \textbf{0.753} & \textbf{0.752}  \\
    \midrule
    \multirow{5}{*}{BioASQ}
      & Plain QA        & 13.48 & 32.02 & 96.84 & 4.313 & 1.000 & 1.000 & 1.000 & 0.000 & 0.000 & 0.000  \\
      & Pr$\varepsilon$$\varepsilon$mpt  & 13.44 & 31.99 & 96.81 & 4.156 & 1.000 & 1.000 & 1.000 & 0.041 & 0.497 & 0.151  \\  \noalign{\vskip 0.3em}
      \cdashline{2-12} \noalign{\vskip 0.5em}
      & InferDPT   & 8.97  & 19.88 & 77.44 & 4.117 & 0.709 & 0.540 & 0.730 & 0.094 & 0.375 & 0.291  \\
      & HAS      & 11.91 & 28.45 & 95.21 & 4.346 & 0.998 & 0.958 & 0.977 & 0.103 & 0.553 & 0.203  \\
      & \textbf{CROSS-MAP}
                        & \textbf{12.73} & \textbf{30.90} & \textbf{96.10}
                        & 5.128 & 0.879 & 0.610 & 0.805
                        & \textbf{0.604} & \textbf{0.858} & \textbf{0.919}  \\
    \midrule
    \multirow{5}{*}{FinCausal}
      & Plain QA        & 59.27 & 77.61 & 81.20 & 4.712 & 1.000 & 1.000 & 1.000 & 0.000 & 0.000 & 0.000  \\
      & Pr$\varepsilon$$\varepsilon$mpt  & 59.20 & 77.54 & 81.15 & 4.722 & 1.000 & 1.000 & 1.000 & 0.016 & 0.150 & 0.037  \\ \noalign{\vskip 0.3em}
      \cdashline{2-12} \noalign{\vskip 0.5em}
      & InferDPT   & 58.63 & 77.10 & 79.06 & 4.824 & 0.964 & 0.893 & 0.949 & 0.054 & 0.238 & 0.036  \\
      & HAS      & 57.40 & 75.89 & 79.24 & 4.645 & 0.968 & 0.921 & 0.956 & 0.092 & 0.294 & 0.139  \\
      & \textbf{CROSS-MAP}
                        & \textbf{58.99} & \textbf{77.33} & \textbf{80.61}
                        & 4.970 & 0.860 & 0.733 & 0.877
                        & \textbf{0.540} & \textbf{0.710} & \textbf{0.678}  \\
    \midrule
    \multirow{5}{*}{NarrativeQA}
      & Plain QA        & 21.12 & 44.21 & 90.68 & 4.518 & 1.000 & 1.000 & 1.000 & 0.000 & 0.000 & 0.000  \\
      & Pr$\varepsilon$$\varepsilon$mpt  & 21.07 & 44.15 & 90.65 & 4.681 & 1.000 & 1.000 & 1.000 & 0.011 & 0.516 & 0.117  \\  \noalign{\vskip 0.3em}
      \cdashline{2-12} \noalign{\vskip 0.5em}
      & InferDPT   & 15.20 & 38.74 & 72.12 & 4.545 & 0.846 & 0.910 & 0.955 & 0.014 & 0.251 & 0.154  \\
      & HAS      & 19.61 & 41.92 & 87.80 & 4.512 & 0.943 & 0.897 & 0.954 & 0.094 & 0.541 & 0.341  \\
      & \textbf{CROSS-MAP}
                        & \textbf{20.49} & \textbf{43.22} & \textbf{89.67}
                        & 5.138 & 0.829 & 0.689 & 0.863
                        & \textbf{0.313} & \textbf{0.796} & \textbf{0.880}  \\
    \bottomrule
  \end{tabular}
  \caption{Privacy and utility results across datasets for Plain QA, Pr$\varepsilon$$\varepsilon$mpt~\citep{chowdhury2025pr}, InferDPT~\citep{tong2025inferdpt}, HAS~\citep{chen2023hideseekhaslightweight}, and CROSS-MAP. The results are reported using Qwen2.5-14B as the base model for both the mapper and the recoverer in CROSS-MAP. An LLAMA-8B model is used for downstream QA.}
  \label{tab:main_results_all}
  \vspace{-10pt}
\end{table*}

In Table~\ref{tab:main_results_all}, the ACC metric calculates the proportion of recovered LLM responses that contain the ground-truth answer span. The similarity (Sim.) metric for recovered text is defined as $1-d_{\text{text}}(x, \tilde{x})$. The BLEU and ROUGE-1~(R-1) metrics are also computed by comparing $x$ with $\tilde{x}$. \rev{Fluency Loss (Flu. Loss) is the token-level negative log-likelihood under the reference LM, i.e., the negative of the fluency score $f_{\mathrm{text}}$ in Section~\ref{subsec:key_components}.} Other metrics are consistent with the definitions in Section~\ref{subsec:key_components}. 

Among all methods, Plain QA and Pr$\varepsilon$$\varepsilon$mpt retain the highest downstream utility but provide negligible or weak privacy. For Pr$\varepsilon$$\varepsilon$mpt, the mapped-text privacy scores remain low overall, reflecting its minimal perturbation on the original text. InferDPT delivers only partial privacy gains while causing a substantial utility collapse (e.g., on SQuAD, F1 drops from 84.10 to 74.04), highlighting that perturbation without an explicit recovery mechanism distorts answer-critical semantics. As a bidirectional framework, HAS attains higher utility than InferDPT, but its privacy remains below CROSS-MAP, suggesting that residual semantic cues are still preserved in the mapped text and thus remain exposable to external LLMs. Across all datasets, CROSS-MAP consistently achieves the strongest mapped-text privacy while largely preserving downstream QA utility relative to other baselines. For example, on SQuAD, the distance between mapped text and original text increases from HAS's 0.238 to 0.752 while QA accuracy increases from HAS's 91.84 to 92.98. Similar trends are observed on other datasets. These consistent gains indicate that high downstream utility does not require high similarity between mapped and original text. In summary, the consistent privacy gains and the near-preserved downstream utility achieved by CROSS-MAP demonstrate that a bidirectional framework with semantic decoupling yields the best privacy-utility trade-off among the compared methods.

\rev{Table~\ref{tab:method_plausibility_judge} reports LLM-as-a-judge ratings across different methods. The judge rates all metrics on a 1--5 scale, where higher scores in the first two columns indicate better plausibility (Plaus.) and fluency (Flu.), whereas lower scores in the last column indicate fewer conspicuous obfuscation artifacts for detectability (Detect.).}

\begin{table}[H]
  \centering
  \footnotesize
  \setlength{\tabcolsep}{3pt}
  \begin{tabular}{lccc}
    \toprule
    Method & Plaus. $\uparrow$ & Fluency $\uparrow$ & Detect. $\downarrow$ \\
    \midrule
    Original text & 4.93 & 4.91 & 1.06 \\
    Pr$\varepsilon$$\varepsilon$mpt & 4.82 & 4.86 & 1.18 \\
    HAS & 4.54 & 4.61 & 1.58 \\
    InferDPT & 3.71 & 3.94 & 2.47 \\
    CROSS-MAP  & 4.45 & 4.52 & 1.72 \\
    % CROSS-MAP (near domain) & 4.45 & 4.52 & 1.72 \\
    % CROSS-MAP (far-compatible) & 4.12 & 4.31 & 2.05 \\
    \bottomrule
  \end{tabular}
  \caption{\rev{LLM-as-a-judge ratings across methods.}}
  \label{tab:method_plausibility_judge}
  \vspace{-10pt}
\end{table}

\vspace{-4pt}
\subsubsection{De-anonymization Attacks}
\vspace{-2pt}
Compared to Plain Text, Pr$\varepsilon$$\varepsilon$mpt achieves only a marginal reduction in attack success ($-1.7\%$) and confidence ($-6.4\%$), indicating that limited edits remain insufficient to remove attribute cues. In contrast, HAS reduces inference accuracy by $34.1\%$, and GPT-AA further reduces it by $40.7\%$, accompanied by lowered certainty. The strongest defense is achieved by CROSS-MAP, where inference accuracy drops significantly by $88.5\%$. This pattern suggests that the semantic decoupling in CROSS-MAP suppresses attribute cues more thoroughly, leading to both lower attacker accuracy and lower confidence.

\begin{table}[t]
  \centering
  \small
  \begin{tabular}{lccc}
    \toprule
    Method & Inference Accuracy & Inference Certainty \\
    \midrule
    Plain Text & 0.713 & 3.366 \\
    Pr$\varepsilon$$\varepsilon$mpt & 0.701 & 3.152 \\
    HAS & 0.470 & 2.196 \\
    GPT-AA     & 0.423 & 2.073 \\
    \textbf{CROSS-MAP}    & \textbf{0.082} & \textbf{1.063} \\
    \bottomrule
  \end{tabular}
  \caption{De-anonymization results via LLM attribute inference on SynthPAI~\citep{yukhymenko2024synthetic}, reporting attacker inference accuracy and confidence (0-5).}
  \label{tab:attr_infer}
  \vspace{-15pt}
\end{table}

% \vspace{-13pt}
\begin{figure*}[t]
  \centering
  \includegraphics[width=\textwidth]{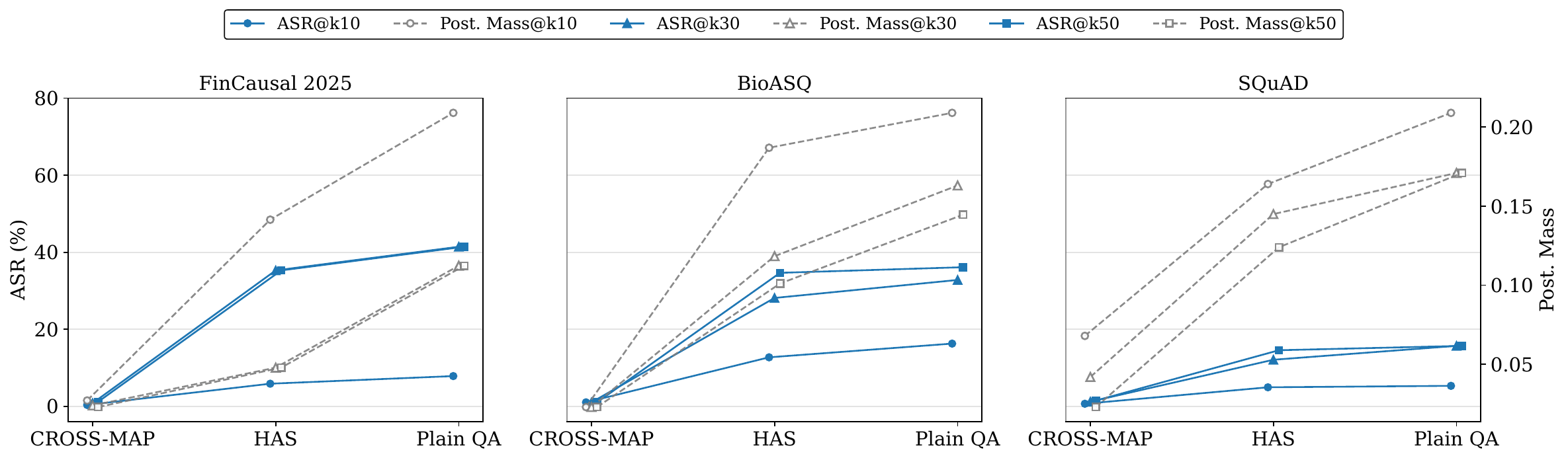}
  \vspace{-18pt}
  \caption{Correlation between span restoration accuracy and the attacker's posterior mass on protected spans.}
  % \vspace{-2pt}
  \label{fig:blackbox_asr_superisal}
\end{figure*}

\begin{figure*}[t]
  \centering
  \includegraphics[width=0.6\textwidth]{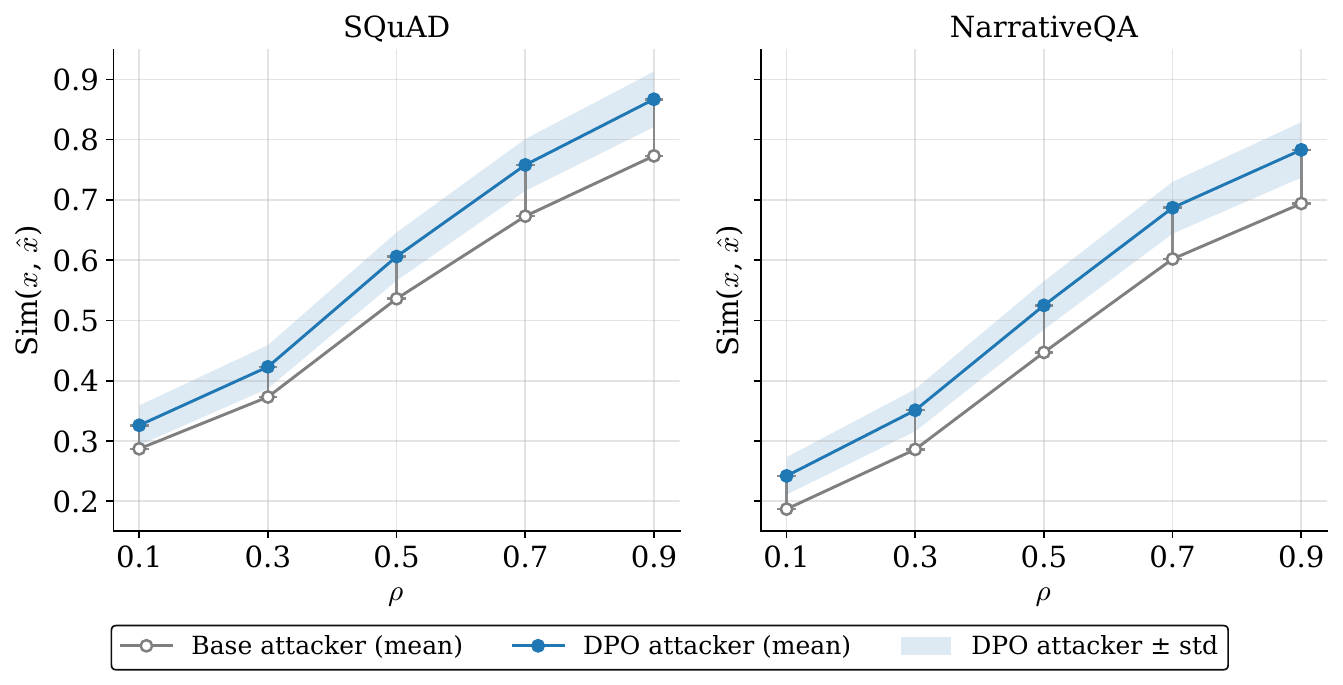}
  \caption{Optimization-based attacker's restoration rate under different dictionary leakage ratios $\rho$.}
  \label{fig:whitebox_asr}
\end{figure*}

\subsubsection{Span Restoration Attacks}
% \vspace{-3pt}

Figure~\ref{fig:blackbox_asr_superisal} provides direct empirical evidence that posterior sampling succeeds when the protected span remains a high-probability continuation under the attacker's token-level posterior. The privacy advantage of CROSS-MAP is therefore attributed to its semantic-decoupling mapping, where direct lexical and semantic anchoring between mapped text and the original protected spans is systematically weakened. As a result, the attacker's posterior over the protected spans is flattened. The true span no longer concentrates probability mass among top-$k$ candidates, yielding lower posterior mass, equivalently higher surprisal, and lower ASR consistently across datasets.

\subsubsection{Optimization-based Attacks}

Table~\ref{tab:rl_attack_squad} and Table~\ref{tab:rl_attack_narrative} show that restoration becomes easier as $\rho$ increases for both datasets, and that the DPO-based attacker consistently outperforms the base attacker without DPO optimization at every $\rho$. For example, on SQuAD, the DPO-based attacker's mean similarity score increases from $0.326$ to $0.867$ as $\rho$ goes from $0.1$ to $0.9$, with absolute mean gains over the base attacker ranging from $0.039$ to $0.094$. This systematic improvement suggests that optimization reliably pushes generations toward better reconstruction, rather than merely benefiting from occasional lucky samples. Qwen2.5-3B is used for both attackers in these two tables.

From Figure~\ref{fig:whitebox_asr}, it can be clearly observed that the restoration similarity rises monotonically with $\rho$. Furthermore, higher leakage is also accompanied by larger absolute gains. A possible reason is that higher $\rho$ reveals more dictionary information, sharpening the attacker's posterior, while DPO further concentrates probability mass on outputs that yield higher restoration similarity.

\subsubsection{Ablation Studies}

This section details an ablation analysis of key model and deployment factors from multiple perspectives: input complexity, paraphrase robustness, target-domain selection, dictionary coverage, training scale, external LLM capabilities, mapper/recoverer size, and training strategies.
\paragraph{\rev{Input complexity.}}
\rev{Three parse-based metrics are used to measure syntactic complexity: mean dependency distance (MDD), maximum dependency-tree depth (max depth), and clauses per sentence (clauses/sent.). NarrativeQA has an average MDD, max depth, and clauses per sentence of $2.40$, $12.8$, and $2.95$, respectively. SQuAD has $2.40/10.6/2.26$, BioASQ has $2.41/7.6/2.03$, and FinCausal has $1.98/7.7/1.93$. 

Compared to the example sentence, \textit{``Although the defendant had no prior convictions, the judge imposed a harsher sentence because the minor victim suffered irreversible harm''} (MDD $=3.00$, max depth $=6$, and clauses/sent. $=4.00$), the datasets used for evaluation in this paper include multi-clause and long-distance syntactic structures.}

\begin{table}[H]
  \centering
  \scriptsize
  \setlength{\tabcolsep}{2.5pt}
  \resizebox{\linewidth}{!}{%
  \begin{tabular}{lccccc}
    \toprule
    Split & MDD & Max depth & Clauses/sent. & ACC $\uparrow$ & Flu. Loss $\downarrow$ \\
    \midrule
    Q1 & 2.16 & 11.11 & 2.27 & 93.62\% & 4.722 \\
    Q2 & 2.40 & 10.75 & 2.34 & 90.52\% & 4.687 \\
    Q3 & 2.65 & 9.85 & 2.53 & 93.36\% & 4.715 \\
    Q4 & 3.24 & 10.11 & 2.90 & 94.42\% & 4.688 \\
    \bottomrule
  \end{tabular}
  }
  \caption{\rev{SQuAD complexity-quartile robustness. Examples are sorted by MDD and split into Q1--Q4, where Q4 is most complex.}}
  \label{tab:syntax_complexity_diagnostics}
  \vspace{-10pt}
\end{table}

\rev{Table~\ref{tab:syntax_complexity_diagnostics} reports the evaluation of complexity on the SQuAD dataset. The results show that increasing syntactic complexity does not monotonically degrade utility or fluency. Q4, the most complex quartile, reaches $94.42\%$ ACC and a Fluency Loss of $4.688$. This suggests that CROSS-MAP is robust to complex structures.}

\paragraph{\rev{Paraphrasing Robustness.}}
\rev{Table~\ref{tab:paraphrase_robustness} evaluates CROSS-MAP's recovery robustness against the paraphrasing of mapped text. The results show that the ACC for SQuAD drops mildly from $0.929$ to $0.877$ under WordNet synonym substitution ($p=0.30$) and to $0.896$ under T5 paraphrasing. Such robustness to paraphrasing is important for third-party LLM APIs, where downstream generation cannot be forced to copy the mapper's terms exactly.}
\begin{table}[H]
  \centering
  \footnotesize
  \setlength{\tabcolsep}{3pt}
  \begin{tabular}{lccc}
    \toprule
    Setting & ACC & Drop \\
    \midrule
    Original output & 0.929 & -- \\
    WordNet synonyms & 0.877 & 5.2\% \\
    T5 paraphrase & 0.896 & 3.3\% \\
    \bottomrule
  \end{tabular}
  \caption{\rev{Paraphrase robustness of the recoverer on SQuAD. The recoverer reasons over the dictionary and response jointly, rather than relying on exact string matching.}}
  \label{tab:paraphrase_robustness}
  \vspace{-10pt}
\end{table}

\paragraph{\rev{Domain selection.}}
\rev{To investigate the interplay between semantic domain distance, privacy preservation, and task utility, we conduct an ablation analysis in Table~\ref{tab:domain_distance_ablation}. The within-cluster row averages $d_{\mathrm{text}}$ and ACC over samples mapped between source-target domain pairs in the same cluster. The across-cluster row averages over samples mapped between pairs from different clusters. The results indicate that across-cluster mappings significantly augment $d_{\mathrm{text}}$ from $0.642$ to $0.797$, confirming that increased domain divergence strengthens semantic separation. Notably, this substantial privacy gain is achieved with only a minor reduction in utility (a $3.5\%$ decrease in accuracy, shifting from $90.4\%$ to $86.9\%$). This desirable trade-off suggests that the mapping can be executed with high quality despite a large semantic distance, provided that there is sufficient structural and relational alignment between the source and target domains to preserve the underlying discourse mechanics.}
\begin{table}[H]
  \centering
  \footnotesize
  \setlength{\tabcolsep}{3pt}
  \begin{tabular}{lccc}
    \toprule
    Pairing & $d_{\mathrm{text}}$ & ACC  \\
    \midrule
    Within-cluster & 0.642 & 90.4\%  \\
    Across-cluster & 0.797 & 86.9\%  \\
    \bottomrule
  \end{tabular}
  \caption{\rev{Target-domain distance ablation over 260 source-target pairs. Source and target domains are embedded with Sentence-BERT and grouped into semantic-domain clusters.}}
  \label{tab:domain_distance_ablation}
  \vspace{-10pt}
\end{table}

\rev{Table~\ref{tab:domain_selection_case} further provides a qualitative case study for domain selection. The same cognitive-development source text is mapped to a far equipment-operations domain and a nearer medical-training domain. Both mappings preserve the key discourse functions: a long-range temporal frame, causal connective, contrastive predication, multi-clause explanation, and comparative grounding. However, the far-domain mapping needs more local reordering to remain natural, while the near-domain mapping can preserve more surface structure. This explains the quantitative pattern in Table~\ref{tab:domain_distance_ablation}: larger semantic distance can improve privacy, but the target domain needs enough compatible relational roles to support a fluent mapped passage and ensure accurate recovery.}
\begin{table*}[t]
  \centering
  \scriptsize
  \setlength{\tabcolsep}{3pt}
  \begin{tabularx}{\textwidth}{p{0.16\textwidth}X}
    \toprule
    Mapping & Text \\
    \midrule
    Original: cognitive development &
    \textbf{Wisdom}, or the capacity for \textbf{insight and judgment} that is developed through \textbf{experience}, increases between the \textbf{ages of fourteen and twenty-five}, then levels off. Thus, it is during the \textbf{adolescence-adulthood transition} that individuals acquire the type of \textbf{wisdom} typically associated with age. \textbf{Wisdom} is not the same as \textbf{intelligence}: adolescents do not improve substantially on \textbf{IQ tests} because their scores are relative to others in their \textbf{same age group}, and relative standing usually does not change --- everyone matures at approximately the same rate in this way. \\
    \midrule
    Far domain: equipment operations &
    \textbf{Diagnostic maturity}, or the ability to \textbf{interpret sensor patterns and make maintenance decisions} refined through \textbf{operating hours}, improves between \textbf{300 and 1,800 operating hours}, then plateaus. Thus, units tend to develop the level of \textbf{diagnostic maturity} typically attributed to long service life during the transition from \textbf{commissioning to routine operation}. \textbf{Diagnostic maturity} is not the same as \textbf{rated capacity}: newly commissioned units do not show large gains on \textbf{standardized bench tests} because their results are normalized against \textbf{comparable models from the same production batch}, and relative ranking typically remains stable --- most units stabilize at roughly similar rates under similar duty cycles on that metric. \\
    \midrule
    Near domain: medical training &
    \textbf{Professional judgment}, or the capacity for \textbf{clinical insight and decision-making} that is developed through \textbf{supervised practice}, increases between the \textbf{first and final years of training}, then levels off. Thus, it is during the \textbf{trainee-practitioner transition} that clinicians acquire the type of \textbf{professional judgment} typically associated with seniority. \textbf{Professional judgment} is not the same as \textbf{medical knowledge}: trainees do not improve substantially on \textbf{standardized exams} because their scores are relative to others in their \textbf{same training cohort}, and relative standing usually does not change --- everyone advances at approximately the same rate in this way. \\
    \bottomrule
  \end{tabularx}
  \caption{\rev{Domain-selection case study. Bold spans mark key source entities and their mapped counterparts.}}
  \label{tab:domain_selection_case}
  \vspace{-10pt}
\end{table*}

\paragraph{\rev{Dictionary coverage.}}
\rev{Table~\ref{tab:dictionary_coverage_ablation} ablates recovery behavior as dictionary coverage decreases. Recovery similarity decreases from $0.860$ at approximately $0.8$ coverage to $0.816$ at approximately $0.5$ coverage, and remains $0.788$ at approximately $0.2$ coverage. The trend is smooth rather than catastrophic because the recoverer can use the mapped response and task context to infer some missing links. However, the ACC drop from $91.2\%$ to $86.2\%$ also shows that the dictionary is not optional. It is the mechanism that keeps semantic decoupling recoverable for downstream use.}

\begin{table}[H]
  \centering
  \small
  \begin{tabular}{lcc}
    \toprule
    Dictionary coverage & Recovery Sim. & ACC \\
    \midrule
    $c_{\mathrm{dict}}\approx0.8$ & 0.860 & 91.2\% \\
    $c_{\mathrm{dict}}\approx0.5$ & 0.816 & 88.7\% \\
    $c_{\mathrm{dict}}\approx0.2$ & 0.788 & 86.2\% \\
    \bottomrule
  \end{tabular}
  \caption{\rev{Dictionary-coverage ablation on SQuAD.}}
  \label{tab:dictionary_coverage_ablation}
  \vspace{-10pt}
\end{table}

\paragraph{\rev{Training scale.}}
\rev{Table~\ref{tab:sft_scaling} ablates SFT data size for the 14B model. Moving from the base model to 200 and 500 SFT examples already improves both privacy distance and utility, indicating that the desired behavior can be learned from a relatively small synthetic set. The largest utility gain appears by 700 examples, where ACC reaches $94.59$ and recovery similarity reaches $0.945$. Additional data from 1000 to 2000 examples yields only marginal changes in $d_{\mathrm{text}}$ and slightly lower recovery similarity, suggesting that later gains are dominated by DPO and sampling strategy rather than by simply scaling SFT data.}

\begin{table}[H]
  \centering
  \small
  \begin{tabular}{lccc}
    \toprule
    Checkpoint & \# samples & $d_{\mathrm{text}}$ & Sim. / ACC \\
    \midrule
    Base & 0 & 0.541 & 0.901 / 88.96 \\
    SFT-200 & 200 & 0.634 & 0.912 / 90.38 \\
    SFT-500 & 500 & 0.651 & 0.933 / 92.57 \\
    SFT-700 & 700 & 0.663 & 0.945 / 94.59 \\
    SFT-1000 & 1000 & 0.667 & 0.943 / 94.34 \\
    SFT-1500 & 1500 & 0.670 & 0.940 / 94.35 \\
    SFT-2000 & 2000 & 0.673 & 0.937 / 94.51 \\
    \bottomrule
  \end{tabular}
  \caption{\rev{SFT training-data scaling study. Gains largely saturate around 700--1000 samples.}}
  \label{tab:sft_scaling}
  \vspace{-10pt}
\end{table}

\paragraph{\rev{External LLMs.}}
\rev{Table~\ref{tab:external_llm_ablation} compares local-only inference, unprotected external-LLM access, and external-LLM access through CROSS-MAP. The SQuAD rows use a 3B local mapper/recoverer and a Qwen-2.5-14B model as the external LLM, showing that CROSS-MAP remains useful even when the external model is a local open-weight model rather than a commercial API. The remaining rows use GPT-5 as the external LLM. In both regimes, CROSS-MAP retains most of the external model's utility while keeping original-domain sensitive content local.}

\begin{figure}[H]
  \centering
  \includegraphics[width=0.95\linewidth]{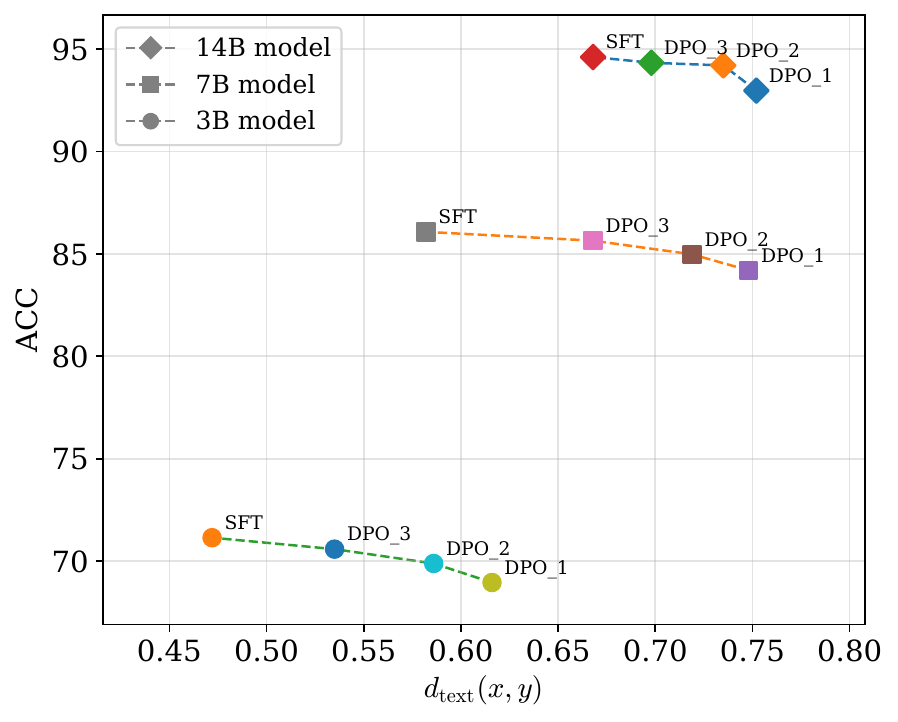}
  \vspace{-8pt}
  \caption{Ablation study of model scale and training objectives on utility--privacy trade-off. Dashed lines connect checkpoints within the same model size. \rev{$\text{DPO}\_3$, $\text{DPO}\_2$, and $\text{DPO}\_1$ denote DPO checkpoints from epochs 1, 3, and 5, respectively, showing increasingly aggressive preference optimization toward privacy-oriented mapping.}}
  \label{fig:squad_ablation_tradeoff}
  \vspace{-10pt}
\end{figure}

\paragraph{Model Size and Training Strategy.} 
Table~\ref{tab:ablation_results} ablates training strategy (SFT-only vs.\ DPO) and mapper~/~recoverer size. Different DPO variants provide adjustable privacy-utility trade-offs for each model size. For the 14B model, as the DPO variant moves from $\text{DPO}_3 \rightarrow \text{DPO}_2 \rightarrow \text{DPO}_1$, privacy distances increase substantially, at the cost of slightly reduced QA utility. The strongest privacy setting $\text{DPO}_1$ (14B) raises $d_{\text{text}}$ to $0.752$ while keeping ACC at $92.98$, i.e., only a $-1.85\%$ drop from Plain QA ($94.73$).

The SFT-only model provides weaker privacy than $\text{DPO}_1$ and a closer approximation of the plain QA quality. However, SFT-only exhibits substantially worse JSON-structure validity $r_{\text{json}}$ than DPO variants, indicating that preference optimization also stabilizes format consistency, which is critical for reliable mapping and recovery. 
Figure~\ref{fig:squad_ablation_tradeoff} visualizes the same model-size and training-stage trade-off, with DPO checkpoints moving toward stronger privacy at small utility cost.

% \clearpage

\begin{table*}[!t]
  \centering
  \scriptsize
  \setlength{\tabcolsep}{3pt}
  \begin{tabular}{lllccc}
    \toprule
    Setting & Dataset & Model / pipeline & EM & F1  \\
    \midrule
    Local only & SQuAD & Qwen-2.5-3B local & 33.82 & 42.77  \\
    External, no protection & SQuAD & Qwen-2.5-14B external & 54.28 & 62.51  \\
    CROSS-MAP & SQuAD & 3B mapper/recoverer + 14B external & 47.71 & 55.96  \\
    \midrule
    Local only & BioASQ & Qwen-2.5-14B local & 12.63 & 23.35  \\
    External, no protection & BioASQ & GPT-5 external & 24.98 & 46.43  \\
    CROSS-MAP & BioASQ & 14B mapper/recoverer + GPT-5 external & 22.36 & 44.57 \\
    \midrule
    Local only & FinCausal & Qwen-2.5-14B local & 48.69 & 60.99 \\
    External, no protection & FinCausal & GPT-5 external & 82.32 & 91.74  \\
    CROSS-MAP & FinCausal & 14B mapper/recoverer + GPT-5 external & 81.57 & 90.26  \\
    \midrule
    Local only & NarrativeQA & Qwen-2.5-14B local & 16.72 & 39.91 \\
    External, no protection & NarrativeQA & GPT-5 external & 48.75 & 72.94  \\
    CROSS-MAP & NarrativeQA & 14B mapper/recoverer + GPT-5 external & 47.12 & 71.27  \\
    \bottomrule
  \end{tabular}
  \caption{\rev{External-LLM ablation. CROSS-MAP preserves most of the stronger external model's utility while keeping original-domain sensitive content local.}}
  \label{tab:external_llm_ablation}
  \vspace{-10pt}
\end{table*}

\begin{table*}[t]
  \centering
  \scriptsize
  \begin{tabular}{crrrrr}
    \toprule
    $\rho$ & Base attacker (mean) & DPO attacker (mean) & Gain (mean) & DPO attacker (best) & Std.\ (mean) \\
    \midrule
    0.1 & 0.287 & 0.326 & 0.039 & 0.397 & 0.033 \\
    0.3 & 0.373 & 0.423 & 0.050 & 0.494 & 0.036 \\
    0.5 & 0.536 & 0.606 & 0.070 & 0.679 & 0.040 \\
    0.7 & 0.673 & 0.758 & 0.085 & 0.831 & 0.043 \\
    0.9 & 0.773 & 0.867 & 0.094 & 0.935 & 0.046 \\
    \bottomrule
  \end{tabular}
  \caption{Performance of optimization-based white-box attacker on SQuAD, measuring similarity between reconstructed text and original text across leakage ratios $\rho$.}
  \label{tab:rl_attack_squad}
\end{table*}

\begin{table*}[t]
  \centering
  \scriptsize
  \begin{tabular}{crrrrr}
  \toprule
  $\rho$ & Base attacker (mean) & DPO attacker (mean) & Gain (mean) & DPO attacker (best) & Std.\ (mean) \\
  \midrule
  0.1 & 0.187 & 0.242 & 0.055 & 0.295 & 0.039 \\
  0.3 & 0.286 & 0.351 & 0.065 & 0.418 & 0.045 \\
  0.5 & 0.447 & 0.525 & 0.078 & 0.607 & 0.049 \\
  0.7 & 0.602 & 0.687 & 0.085 & 0.772 & 0.051 \\
  0.9 & 0.694 & 0.783 & 0.089 & 0.832 & 0.046 \\
  \bottomrule
  \end{tabular}
  \caption{Performance of optimization-based white-box attack on NarrativeQA, measuring similarity between reconstructed text and original text across leakage ratios $\rho$.}
  \label{tab:rl_attack_narrative}
  \end{table*}

\begin{table*}[t]
  \centering
  \scriptsize
  \setlength{\tabcolsep}{3.5pt}
  \begin{tabular}{lccc|cccc}
    \toprule
    & \multicolumn{3}{c|}{\textbf{Downstream QA utility}}
    & \multicolumn{4}{c}{\textbf{Mapping quality}} \\
    \cmidrule(lr){2-4}\cmidrule(lr){5-8}
    Method
    & {EM} & {F1} & {ACC}
    & {Sim.} & {$d_{\text{dict}}$} & {$d_{\text{text}}$} & {$r_{\text{json}}$} \\
    \midrule
    Plain QA        & 76.4  & 84.10 & 94.73 & 1.000 & 0.000 & 0.000 & 1.000 \\
    \noalign{\vskip 0.3em}
    \cdashline{1-8}
    \noalign{\vskip 0.5em}

    \textbf{DPO\_1 (14B)}
                    & \textbf{74.94} & \textbf{83.22} & \textbf{92.98} & \textbf{0.826} & \textbf{0.753} & \textbf{0.752} & \textbf{1.000} \\
    DPO\_2 (14B)    & 75.30 & 83.31 & 94.21 & 0.892 & 0.711 & 0.735 & 1.000 \\
    DPO\_3 (14B)    & 75.72 & 83.75 & 94.34 & 0.898 & 0.704 & 0.698 & 1.000 \\
    \textbf{SFT-only (14B)}  & \textbf{75.88} & \textbf{83.82} & \textbf{94.62} & \textbf{0.948} & \textbf{0.612} & \textbf{0.668} & \textbf{0.936} \\
    Base Model (14B)& 69.85 & 78.90 & 88.96 & 0.901 & 0.566 & 0.541 & 0.988 \\

    \noalign{\vskip 0.3em}
    \cdashline{1-8}
    \noalign{\vskip 0.5em}

    \textbf{DPO\_1 (7B)}
                    & \textbf{70.21} & \textbf{81.88} & \textbf{84.20} & \textbf{0.829} & \textbf{0.758} & \textbf{0.748} & \textbf{1.000} \\
    DPO\_2 (7B)     & 70.74 & 82.31 & 84.98 & 0.844 & 0.736 & 0.719 & 1.000 \\
    DPO\_3 (7B)     & 71.86 & 83.05 & 85.65 & 0.861 & 0.708 & 0.668 & 1.000 \\
    \textbf{SFT-only (7B)}   & \textbf{71.20} & \textbf{82.74} & \textbf{86.08} & \textbf{0.901} & \textbf{0.604} & \textbf{0.582} & \textbf{0.918} \\
    Base Model (7B) & 62.14 & 76.82 & 73.44 & 0.872 & 0.522 & 0.446 & 0.972 \\

    \noalign{\vskip 0.3em}
    \cdashline{1-8}
    \noalign{\vskip 0.5em}

    \textbf{DPO\_1 (3B)}
                    & \textbf{54.21} & \textbf{69.35} & \textbf{68.95} & \textbf{0.775} & \textbf{0.624} & \textbf{0.616} & \textbf{1.000} \\
    DPO\_2 (3B)     & 54.88 & 69.87 & 69.88 & 0.788 & 0.603 & 0.586 & 1.000 \\
    DPO\_3 (3B)     & 56.42 & 70.96 & 70.58 & 0.803 & 0.575 & 0.535 & 1.000 \\
    \textbf{SFT-only (3B)}   & \textbf{55.69} & \textbf{70.34} & \textbf{71.14} & \textbf{0.846} & \textbf{0.401} & \textbf{0.472} & \textbf{0.782} \\
    Base Model (3B) & 46.10 & 64.88 & 60.73 & 0.781 & 0.365 & 0.322 & 0.958 \\
    \bottomrule
  \end{tabular}
  \caption{Ablation study on SQuAD, comparing DPO-tuned and SFT-only models of different sizes in terms of QA utility, privacy distances, and JSON-structure validity. LLAMA-8B model is used for downstream QA.}
  \label{tab:ablation_results}
\end{table*}

\begin{figure*}
  \centering
  \includegraphics[width=\textwidth]{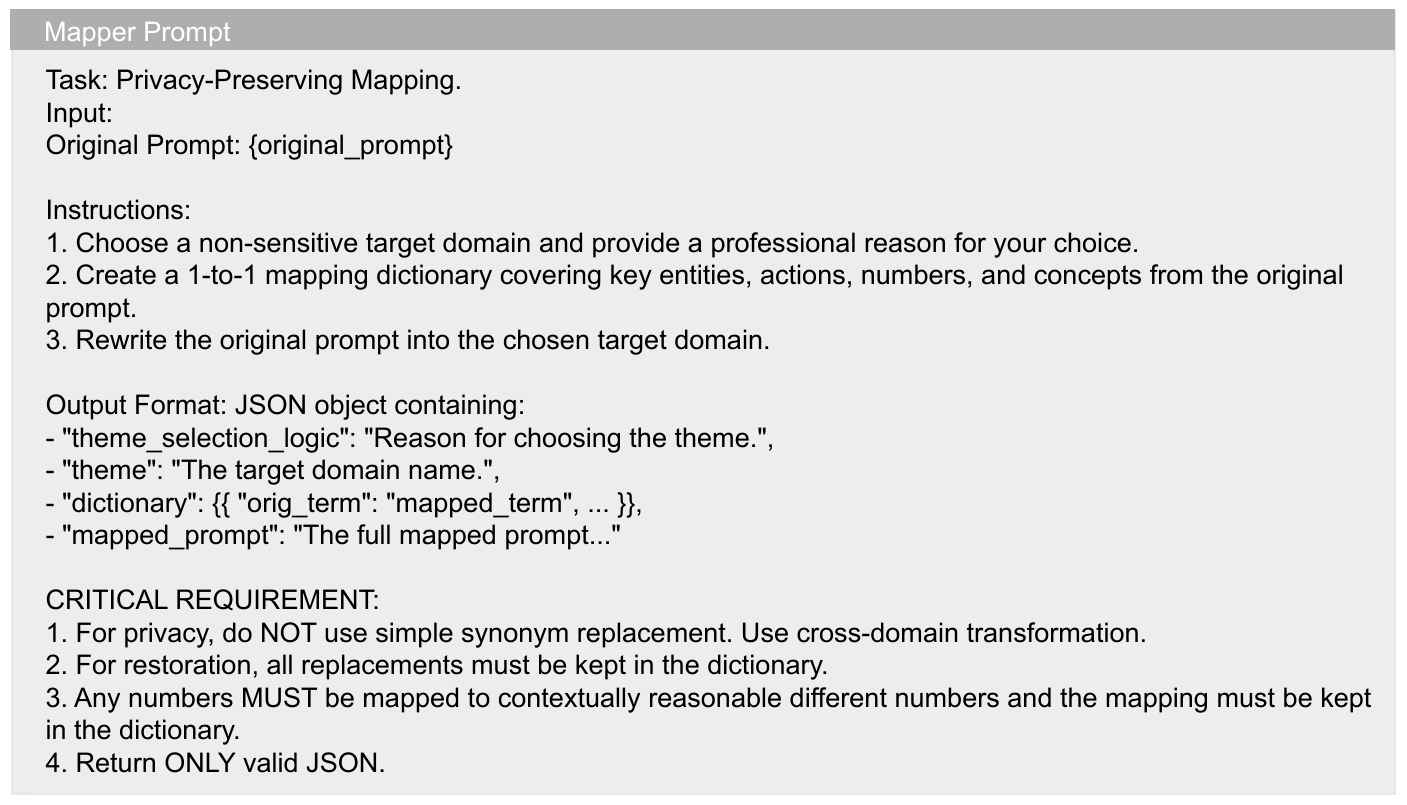}
  \caption{Mapper prompt.}
  \label{fig:mapper_prompt}
\end{figure*}

\begin{figure*}
  \centering
  \includegraphics[width=\textwidth]{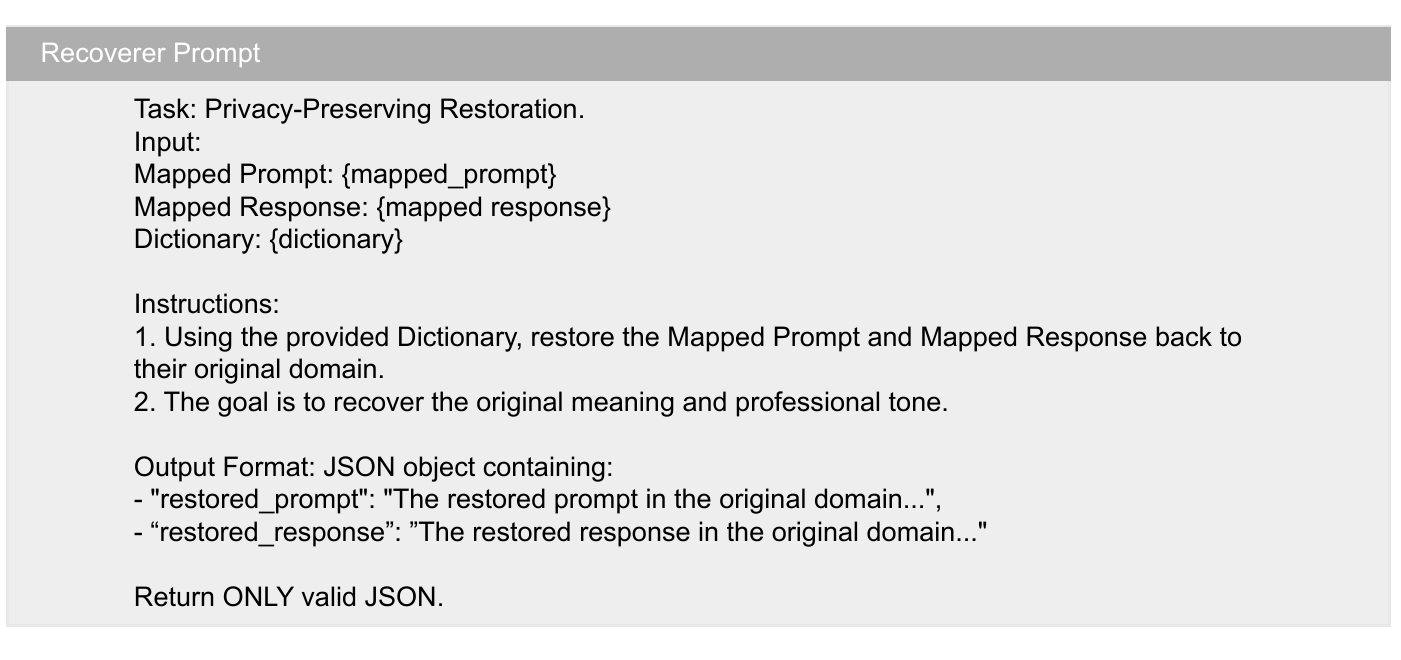}
  \caption{Recoverer prompt.}
  \label{fig:recoverer_prompt}
\end{figure*}

\begin{table*}[t]
  \centering
  \setlength{\tabcolsep}{4pt}
  \renewcommand{\arraystretch}{1.15}
  \footnotesize
  \begin{tabularx}{\textwidth}{
  >{\raggedright\arraybackslash}p{0.15\textwidth}
  >{\raggedright\arraybackslash}p{0.26\textwidth}
  >{\raggedright\arraybackslash}p{0.26\textwidth}
  >{\raggedright\arraybackslash}p{0.26\textwidth}
  }
  \toprule
  Method & Original text & Mapped text & Recovered text  \\
  \midrule
  
  CROSS-MAP & \textcolor{blue}{Fiscal pressures} on \textcolor{blue}{Government} as a result of \textcolor{blue}{reduced revenues} due to \textcolor{blue}{lower oil prices} is evident in the increasing number of \textcolor{blue}{audits} being undertaken to identify areas of \textcolor{blue}{non-compliance}. The \textcolor{blue}{Group} remains committed to \textcolor{blue}{maintaining the highest levels of compliance} and works closely with \textcolor{blue}{local regulatory authorities}.
  &
  \textcolor{green}{Operational strain} on the \textcolor{green}{hospital network} as a result of \textcolor{green}{declining patient volume} due to \textcolor{green}{reduced insurance reimbursements} is evident in the increasing number of \textcolor{green}{quality assurance reviews} being undertaken to identify \textcolor{green}{deviations from clinical protocols}. The \textcolor{green}{hospital network} remains committed to \textcolor{green}{adhering strictly to regulatory standards} and works closely with \textcolor{green}{health department oversight bodies}.
  &
  \textcolor{blue}{Fiscal pressures} on the \textcolor{blue}{Government} as a result of \textcolor{blue}{reduced revenues} due to \textcolor{blue}{lower oil prices} is evident in the increasing number of \textcolor{blue}{audits} being undertaken to identify \textcolor{blue}{non-compliance} with regulatory standards. The \textcolor{blue}{Government} remains committed to \textcolor{blue}{maintaining the highest levels of compliance} and works closely with \textcolor{blue}{local regulatory authorities}.  \\
  \midrule
  HAS~\citep{chen2023hideseekhaslightweight} & Fiscal pressures on \textcolor{blue}
  {Government} as a result of reduced revenues due to lower oil prices is evident in the increasing number of audits being undertaken to identify areas of non-compliance. The \textcolor{blue}{Group} remains committed to maintaining the highest levels of compliance and works closely with local regulatory authorities. & Fiscal pressures on \textcolor{green}{Administration} as a result of reduced revenues due to lower oil prices is evident in the increasing number of audits being undertaken to identify areas of non-compliance. The \textcolor{green}{Committee} remains committed to maintaining the highest levels of compliance and works closely with local regulatory authorities. & Fiscal pressures on \textcolor{blue}{Government} as a result of reduced revenues due to lower oil prices is evident in the increasing number of audits being undertaken to identify areas of non-compliance. The \textcolor{blue}{Committee} remains committed to maintaining the highest levels of compliance and works closely with local regulatory authorities. \\
  \midrule 

  InferDPT~\citep{tong2025inferdpt} & In \textcolor{blue}{2018} the majority of employees across the \textcolor{blue}{Group} have received average salary \textcolor{blue}{increases} ranging from 2.0$\%$–3.25$\%$, dependent on geographical location with the principal exception being those employees based in Brazil, Latin America and China where, due to inflation, current market salary increases are higher. & In \textcolor{green}{2019} the majority of employees across the \textcolor{green}{Company} have received average salary \textcolor{green}{rises} ranging from 2.0$\%$–3.25$\%$, dependent on geographical location with the principal exception being those employees based in Brazil, Latin America and China where, due to inflation, current market salary rises are higher. & In \textcolor{blue}{2019} the majority of employees across the \textcolor{blue}{Company} have received average salary \textcolor{blue}{rises} ranging from 2.0$\%$–3.25$\%$, dependent on geographical location with the principal exception being those employees based in Brazil, Latin America and China where, due to inflation, current market salary rises are higher. \\
  \midrule 

  Pr$\varepsilon$$\varepsilon$mpt~\citep{chowdhury2025pr} & The Committee also reviewed the impact of the reduction in US federal tax rates as a result of tax reform in the US, which resulted in a reduction of deferred tax balances of \textcolor{blue}{$617$} million.
  &
  The Committee also reviewed the impact of the reduction in US federal tax rates as a result of tax reform in the US, which resulted in a reduction of deferred tax balances of \textcolor{green}{$593$} million.
  &
  The Committee also reviewed the impact of the reduction in US federal tax rates as a result of tax reform in the US, which resulted in a reduction of deferred tax balances of \textcolor{blue}{$617$} million. \\
  \bottomrule
  \end{tabularx}
  \caption{Examples of different mapping methods.}
  \label{tab:mapping_examples}
  \end{table*}

\end{document}